\documentclass[manyauthors]{fundam}
\usepackage{graphicx}
\usepackage{float}
\restylefloat{table}
\restylefloat{figure}
\usepackage{algorithm}
\usepackage{algpseudocode}
\usepackage{array}
\usepackage{multirow}
\usepackage{balance}
\usepackage{xcolor}
\usepackage{epstopdf}
\usepackage{caption}
\usepackage{subcaption}

\newtheorem{observation}{Observation}[section]
\newcommand{\tabincell}[2]{\begin{tabular}{@{}#1@{}}#2\end{tabular}}

\begin{document}

\setcounter{page}{1}
\publyear{2021}
\papernumber{0001}
\volume{178}
\issue{1}
\title{Efficient Reporting of Top-$k$ Subset Sums}

\author{Biswajit Sanyal\\Department of Information Technology\\Govt. College of Engg. \& Textile Technology\\Serampore, Hooghly, West Bengal 712 201, India\\
biswajit\_sanyal@yahoo.co.in \and Subhashis Majumder\corresponding\\Department of Computer Science and Engineering\\Heritage Institute of Technology, Kolkata, West Bengal 700 107, India\\ 
subhashis.majumder@heritageit.edu \and Priya Ranjan Sinha Mahapatra\\Department of Computer Science \& Engineering\\University of Kalyani, West Bengal, India\\
priya@klyuniv.ac.in}

\maketitle
\address{Subhashis Majumder, Department of Computer Science and Engineering, Heritage Institute of Technology, Kolkata, West Bengal 700 107, India, subhashis.majumder@heritageit.edu}
\runninghead{B. Sanyal, S. Majumder, P. R. Sinha Mahapatra}{Efficient Reporting of Top-$k$ Subset Sums}

\begin{abstract}

The ``Subset Sum problem" is a very well-known NP-complete problem. In this work, a top-$k$ variation of the ``Subset Sum problem" is considered. This problem has wide application in recommendation systems, where instead of $k$ best objects the $k$ best subsets of objects with the lowest (or highest) overall scores are required. Given an input set $R$ of $n$ real numbers and a positive integer $k$, our target is to generate the $k$ best subsets of $R$ such that the sum of their elements is minimized. Our solution methodology is based on constructing a {\em metadata structure} $G$ for a given $n$. Each node of $G$ stores a bit vector of size $n$ from which a subset of $R$ can be retrieved. Here it is shown that the construction of the whole graph $G$ is not needed. To answer a query, only implicit traversal of the required portion of $G$ on demand is sufficient, which obviously gets rid of the preprocessing step, thereby reducing the overall time and space requirement. A modified algorithm is then proposed to generate each subset incrementally, where it is shown that it is possible to do away with the explicit storage of the bit vector. This not only improves the space  requirement but also improves the asymptotic time complexity. Finally, a variation of our algorithm that reports only the top-$k$ subset sums has been compared with an existing algorithm, which shows that our algorithm performs better both in terms of time and space requirement by a constant factor.

%\keywords{Combination, Metadata Structure, DAG, Partition Tree, Aggregation Function}
\end{abstract}
%\end{frontmatter}

\section{Introduction}
In many application domains, retrieval of the most relevant data items benefits the end users much more than reporting a (potentially huge) list of all the data items that satisfy a certain query. Here the application of aggregation functions to the query results plays a very important role. One of the simplest functions is the top-$k$ aggregation, which reports the $k$ independent objects with the highest scores.

However, instead of a list of $k$ best independent objects, many applications in recommendation systems require $k$ best subsets of objects with lowest (or highest) overall scores. For example, consider an online shopping site with an inventory of items. Obviously, each item has its own cost. Suppose a buyer wants to buy multiple items from the site but he has his own budget constraint. Then recommending a list of $k$ best subsets of items with the lowest overall costs will be helpful to the buyers wherefrom they can pick the subset of items, they require most. In this paper, the above problem is modeled as the top-$k$ subset sums problem that generates the $k$ best subsets of items from an input inventory of items $I$, where a subset with a lower sum of costs occupies a higher position in the top-$k$ list. 

Let us consider another example of trip selection where a visitor wants to visit different places in a continent. As it is known that a continent has several places to visit and obviously each visit has a cost involvement also. So if the visitor has a budget constraint then obviously he can't cover all the places. In that case also, recommending a list of $k$ best subsets of places with lowest overall costs will be helpful to the visitor where from he can pick the subset of places, he admires most. This problem can also be modeled as a top-$k$ subset sums problem that reports $k$ best subsets of places with lowest overall costs.

\subsection{Problem Formulation}
Given a finite set $R$ of $n$ real numbers, $\{ r_1, r_2, \ldots, r_n\}$, sorted in non-decreasing order, our goal is to generate the $k$ best subsets (top-$k$ Subset Sums) for any input value $k$, ranked on the basis of summation function $F$, such that $F(S) = \sum_{r \in S} r$, for any subset $S\subseteq R$. Clearly $\mid S \mid$ $\in [1 \ldots n]$. In our problem, a subset $S_i$ is ranked higher than a subset $S_j$ if $F(S_i) < F(S_j)$. Furthermore, it is assumed that the rank is unique, so that when $F(S_i) = F(S_j)$, ties are broken arbitrarily. Note that if the input set of numbers does not come as sorted, an additional $O(n \log_2 n)$ time can be taken to sort it first. However, since $k < n$ makes the problem trivial, each of the generated subsets being of cardinality one, the $n \log_2 n$ term is typically not mentioned even if the input set does not come as sorted.

\subsection{Past Work} 
Top-$K$ query processing has a rich literature in many different domains, including information retrieval~\cite{bsa_acm94}, databases~\cite{mbg_acm04}, multimedia~\cite{cgm_ieee04}, business analytics~\cite{acdg_bcidsr03}, combinatorial objects~\cite{sta_idea11}, data mining~\cite{gd_sigkdd05}, or computational geometry~\cite{abz_soda11,kn_soda11,rgjr_svl11}. There are also other extensions of top-$k$ queries in other environments, such as no sorted access on restricted lists~\cite{bgm_icde02, ch_sigmod02}, ad hoc top-$k$ queries~\cite{lci_sigmod06} or
no need for exact aggregate scores~\cite{iae_vldb02}. 

 The subset-sum problem is a well-known NP-complete problem~\cite{gj_book79}, which asks whether there exists a subset $S'$ of a given set of integers $S$, whose elements sum to a given target $t$. A lot of its variants are also computationally hard, as for example when the integers are restricted to be only positive. However, the Top-$k$ version that we are dealing with can be solved in polynomial time as long as $k = O(n^c)$, where $n$ is the cardinality of $S$ and $c$ is a constant. However, if we have to report all the subsets of $S$, naturally the time required will be exponential in $n$. Typically, two different variations of the Top-$k$ Subset Sums problem are found in the literature. Some of them generate only the subset sums in the correct order whereas others report the respective subsets also along with their sums. Clearly the latter variation will need a little bit of higher resource in terms of time and/or space. Sanyal et al.~\cite{smhg_tcs20} developed algorithms for reporting all the top-$k$ subsets (top-$k$ combinations), where the subsets are of a fixed size $r$. Their proposed algorithm runs in $O(rk+k\log_2 k)$ time and some of its variants run in $O(r+k\log_2 k)$ time.

In the last few years, many programmers as well as researchers have been attracted to the problem of finding the sum of a particular subset whose rank is $k$, basically a variation of the Top-$k$ Subset Sums problem. Different solutions were proposed for reporting the top-$k$ subset sums. However, the most promising one amongst them appeared to be a $O(k \log_2 k)$ algorithm~\cite{d_math15} proposed by Eppstein. It uses a min-heap and a simple procedure for generating two new subsets from a subset that got extracted from the heap and then inserts these two new subsets into the heap. It first keeps the numbers in a sorted array of ascending order. At each step, given a nonempty subset of array-indices $S$, the algorithm defines the children of $S$ to be $(S - \{max(S) \}) \cup \{max(S) + 1 \}$ and $S \cup \{max(S) + 1 \}$. Note that the first child has the same number of indices as $S$ and the second one has just one more than its parent. Starting with the subset $\{1\}$ as the root node that corresponds to the singleton set with the smallest element from the original set, the child relation continuously inserts new subsets into the min-heap. It can be shown that the algorithm is capable of generating every nonempty subset of positive integers $(1 \ldots n)$. So the generation of the subset of indices in correct order is guaranteed by the above claim and the modus operandi of a min-heap. The sum of the elements belonging to each subset can be easily calculated and reported at every step. In this technique~\cite{d_math15}, each node of the heap needs to maintain two values--(i) a pointer to the maximum index and (ii) the corresponding subset sum. It can report all the top-$k$ subset sums in order, as and when they get generated or if needed only the sum of the $k^{th}$ subset. However, if it has to report the $k^{th}$ subset or as a matter of fact all the $k$ subsets, some extra pointers need to be stored in each node and some additional computation needs to be done as well.

Very recently, in the database domain, Deep et al.~\cite{dk_icdt21} worked on a similar problem. Here ranked enumeration of Conjunctive Query $(CQ)$ results were used to enumerate the tuples of $Q(D)$ according to the order specified by a rank function $rank$. The variable $Q(D)$ was used to denote the result of the query $Q$ over an input database $D$. Their proposed algorithm works in two phases~\cite{dk_icdt21}: a preprocessing phase that builds a data structure (basically a priority queue) and an enumeration phase that outputs $Q(D)$ according to the order specified by $rank$, using the data structure constructed in the preprocessing phase. The problem considered in this manuscript can also be solved using their approach by considering it as a full union of Conjunctive Query (UCQ) $\phi = \phi_1 \cup . . . \phi_n$, where $n$ is the size of the input set $R$, i.e., $\mid R \mid$. The solution will require a preprocessing time of $O(n^{subw+1} \log_2 n )$ and a delay of $O(n \log_2 n)$, where subw is the sub modular width~\cite{d_jacm13} of all decompositions across all CQs $\phi_i$.
\subsection{Our Contribution}

In this manuscript an efficient output sensitive algorithm is first proposed to report the Top-$k$ subset sums along with their subsets, where the size of the subsets $s$ can be anything between {\tt 1} and {\tt n}, with an overall running time of $O(nk + k \log_2 k)$ and we then improve it to $O(k \log_2 k)$. Both the algorithms were implemented and their runtimes were compared on randomly generated test cases. Another version of our algorithm is considered that reports only the top-$k$ subset sums without the subsets, which also runs in $O(k \log_2 k)$ time. It is further shown that, on a large number of problem instances with the inputs varying from small values of $n$ and $k$ to very large ones, our approach consistently performs better than a prior solution~\cite{d_math15} in terms of time and peak memory used, which means though the asymptotic time complexities are the same, the constant factor in our algorithm is definitely less than the earlier work.

\section{Outline of our Technique}

Our solution is based on constructing an implicit {\em metadata structure} $G$. The novelty of our work is that $G$ is never constructed explicitly, rather at run time, just the required portion of $G$ is generated on demand, which obviously saves the high time and space requirement of the preprocessing step.
The paper is organized as follows. In Section~\ref{sec:prob}, first $n$ {\em local metadata structures} $G_1$ to $G_n$ are introduced and it is shown that how they can be used to construct the full {\em metadata structure} $G$. In addition, it is further shown that how $G$ can be used in conjunction with a min-heap structure $H$ to obtain the desired top-$k$ subsets. In this section, we also highlight the problem of duplicate entries in heap $H$ and show how we can remove this problem by modifying the construction of the $G$. 
Ultimately, a modified $G$ is constructed that can report the desired top-$k$ subsets efficiently. Section~\ref{sec:prob} is concluded by showing that to answer a query, the required portions of $G$ can be generated on demand, so that the requirement for creating $G$ in totality is never needed as a part of preprocessing. Two different variations of the algorithm are presented, the latter version being an improvement over the former both in terms of time and space requirement. In Section~\ref{sec:result} the results of our implementation are presented and it is shown how the required runtime varies with different values of $n$ and $k$ for both algorithms. In the later part of Section~\ref{sec:result}, our first algorithm is slightly modified to report only the top-$k$ subset sums and compare our solution with an existing solution~\cite{d_math15}. Both methods are implemented and run under exactly the same inputs and it is shown that our algorithm is consistently performing better than the existing algorithm. Finally, in Section~\ref{section:conclusion}, the article is concluded and some open problems are mentioned.

\section{Generation of top-$k$ subsets}

\label{sec:prob}

In this section, we first consider the following-- given any input set $R$ of $n$ real numbers, and a positive integer $k$, we construct a {\em metadata structure} $G$ on demand to report the top-$k$ subsets efficiently. Here it is assumed that the numbers of the input set $R$ are kept in a list $R'=(r'_1, r'_2, \ldots, r'_n)$, sorted in non-decreasing order and let $P = \{1, 2, \ldots, n\}$ be the set of positions of the numbers in the list. A subset $S \subseteq R$ is now viewed as a sorted list of $\mid S \mid$ distinct positions chosen from $P$. 

\subsection{The {\em metadata structure} $G$}
The {\em metadata structure} $G$ is constructed as a layered Directed Acyclic Graph (DAG), $G=(V, E)$, in a fashion similar to an earlier work~\cite{smhg_tcs20}, where each node $v \in V$ contains the information of $\mid S \mid$ positions of a subset $S \subseteq R$. In DAG $G$, for each node, the $\mid S \mid$ positions are stored as a bit vector $B[1 \ldots n]$. Note that the bit vector $B$ has in total $\mid S \mid$ numbers of 1s and $n$ $-$ $\mid S \mid$ numbers of 0s. The bit value $B[i]=1$ represents that $r'_i$ of $R'$ is included in the subset $S$ whereas $B[i]=0$ says that ${{r'}_i}$ is not in $S$. Consider the bit vector {\tt 110100} for any subset $S$. It says that the $1^{st}$, $2^{nd}$, and $4^{th}$ numbers of the list $R'$, are included in the subset $S$, where the total number of numbers in $R$ is six. 

The directed edges between the nodes of $G$ are drawn using the concept of ``One Shift" as introduced by Sanyal et al.~\cite{smhg_tcs20}. 
In this current work, for each subset $S$, two different variants of ``One Shift" - ``Static One-shift" and ``Incremental One-shift" are considered.

The first variant is something similar to the earlier concept~\cite{smhg_tcs20}, where $S$ and $S'$ are two subsets with $\mid S \mid$ = $\mid S '\mid$ and $S'$ is obtained from $S$ by applying a one shift. This one shift is named as the ``Static One Shift", in order to distinguish it from the other variant. Let $v(S)$ and $v(S')$ be the nodes corresponding to the subsets $S$ and $S$' and note that their bit vector representations contain the same number of 1s. The formal definition is given below. 

\begin{definition}[Static One Shift]
Let $P_{ S}=(p_1, p_2, \ldots, p_{\mid S \mid})$ denote the list of sorted positions of the numbers in a subset $S \subseteq R$ and let $P_{ S' }=(p'_1, p'_2, \ldots, p'_{\mid S' \mid})$ denote the list of sorted positions of the numbers in another subset $S' \subseteq R$ where $\mid S' \mid$ = $\mid S \mid$. Now, if for some $j$, $p'_j = p_j + 1$ and $p'_i = p_i$ for $i \neq j$, then, it is said that $S'$ is a \emph{Static One Shift} of $S$. 
\end{definition}

The second variation is somewhat different and it is named as the ``Incremental One Shift" where the subset $S'$ is a one shift of the subset $S$ and $\mid S' \mid$ = $\mid S \mid$ + 1, i.e., $S'$ contains one more element than $S$. Hence, the bit vector representation of the node $v(S')$ has one extra $1$ than $v(S)$. It is formally defined below.

\begin{definition}[Incremental One Shift]
Let $P_{ S}=(p_1, p_2, \ldots, p_{\mid S \mid})$ denote the list of sorted positions of the numbers in a subset $S \subseteq R$ and
let $P_{ S' }=(p'_1, p'_2, \ldots, p'_{\mid S' \mid})$ denote the list of sorted positions of the numbers in another subset $S' \subseteq R$, where $\mid S' \mid$ = $\mid S \mid$ $+$ $1$. Now, if $\forall i$, $1 \le i \le {\mid S \mid}$, the position index $p_i$ is also present in $P_{S'}$ but for some $j$, $1 \leq j \le {\mid S' \mid}$, the position index $p'_j$ is not in $P_S$, then, we say that $S'$ is an \emph{Incremental One Shift} of $S$.
\end{definition}

Now, let us consider an example of the above two types of shifts. Let, $R=\{3, 7, 12, 14, 25, 45, 51\}$, $S=\{3, 12, 45, 51\}$ in both the cases and $S'=\{3, 14, 45, 51\}$ in the static case and $S'=\{3, 12, 25, 45, 51\}$ in the incremental case. Note that, in the static case, $P_{S}=(1, 3, 6, 7)$ and $P_{ S' }=(1, 4, 6, 7)$ and $S'$ is a \emph{Static One Shift} of $S$, since $p'_2 = p_2 + 1$ and $p'_i = p_i$ for $i \neq 2$. In the latter case, $P_{ S' }=(1, 3, 5, 6, 7)$. Here, $S'$ is a \emph{Incremental One Shift} of $S$, since $\forall i$, $1 \le i \le 4$, the position index $p_i$ is also present in $P_{S'}$ but the position index $p'_3=5$ is not present in $P_S$.

\subsection{Construction of {\em local metadata structures} $G_1$ to $G_n$ using Static One Shift}
Note that if $S$ is a non-empty subset of $R$, then $\mid S \mid$ $\in \{1..n\}$. Using the concept of ``Static One Shift", we first construct a {\em local metadata structure} for each possible size of the subset $S$ and name this {\em local metadata structure} as $G_{\mid S \mid}$. So, $G_{\mid S \mid}$ is basically a directed acyclic graph where each subset of size ${\mid S \mid}$ is present exactly once in some node of the graph and in its bit vector representation, the number of $1$s is also ${\mid S \mid}$. Let us consider the node corresponding to any subset $S$ be $v(S)$. Then there will be a directed edge from node $v(S)$ to node $v(S')$ iff the subset $S'$ is a static one shift of the subset $S$. Clearly we will have $n$ such {\em local metadata structures} 
$G_1$ to $G_n$. The {\em metadata structure} $G_i$ for any fixed value $i$ can be used to generate the top-$k$ subsets of size $i$ efficiently, ranked on the basis of summation function $F$.

To faciliate the process, we maintain a min-heap $H_i$ to store the candidate subsets of size $i$. Clearly there will be $n$ such local min-heaps $H_1$ to $H_n$. Initially, we insert the root node of the metadata structure $G_i$ as the only element in the heap $H_i$. Then to report the top-$k$ subsets of size $i$, at each step, we extract the minimum element $Z$ of $H_i$ and output it as an answer, and then insert all its children $X$ from $G_i$ into $H_i$ with key value $F(X)$. But the problem that we face is that some children $X$ may be present in $H_i$ already, as $X$ may be a static one-shift of more than one nodes in $G_i$. To avoid this problem of duplicate entries in heap $H_i$, we use a technique that is similar to `mandatory one shift' as introduced by Sanyal et al.~\cite{smhg_tcs20} and we name it as the ``Mandatory Static One Shift".

\begin{figure}[t]
\begin{center}
\scalebox{.45}{\includegraphics{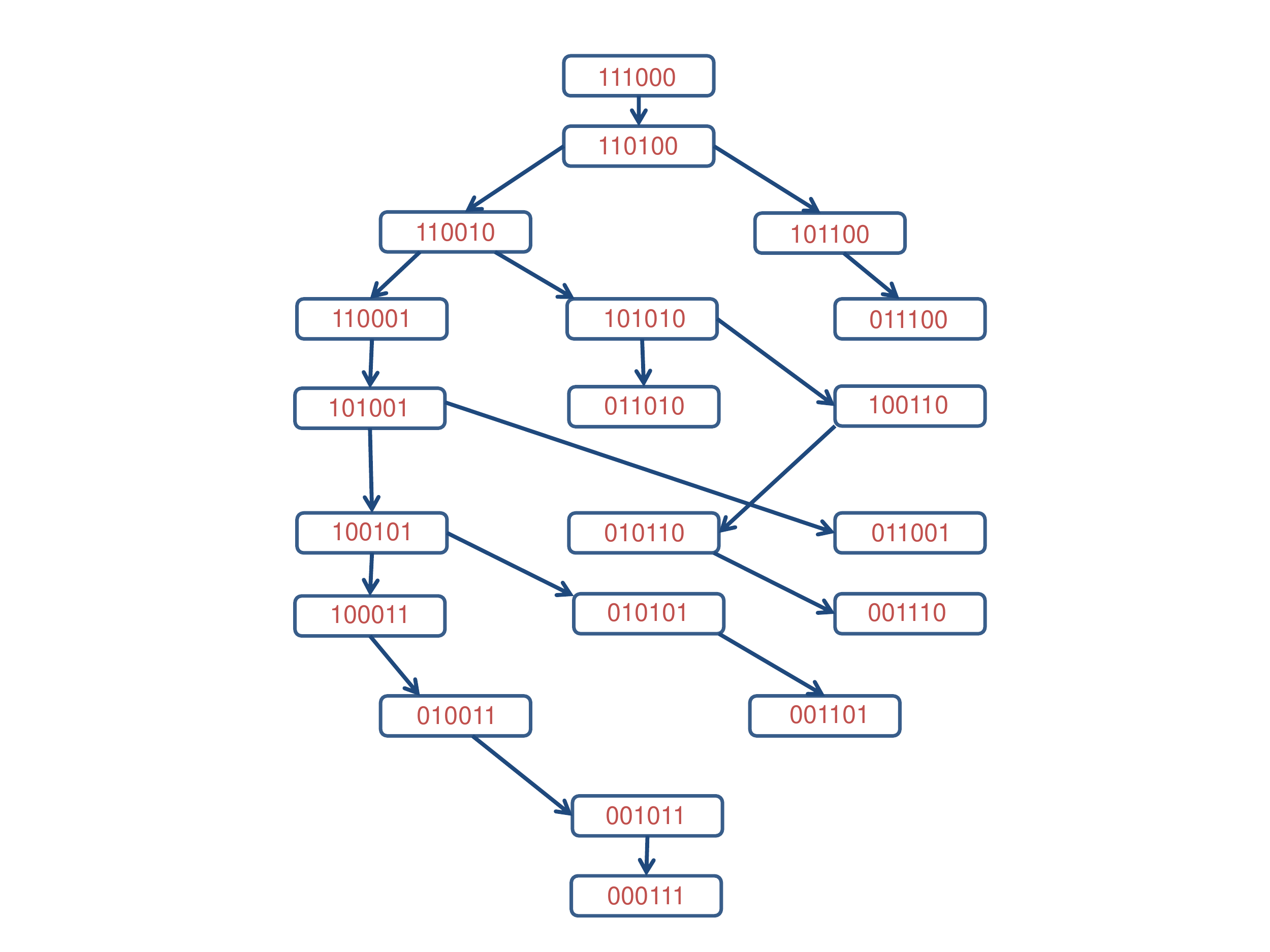}}
\end{center}
\caption{The {\em metadata structure} $G_3$ for the case $n=6$ with mandatory static one shift}
\label{fig:1}
\end{figure}

\begin{definition}[Mandatory Static One Shift]
$S'$ is said to be a \emph{Mandatory Static One Shift} of $S$, if (i) $S'$ is a static one-shift of $S$, and (ii) among all subsets of whom $S'$ is a static-one-shift, $S$ is the 
one whose list of positions is lexicographically the smallest
(equivalently, the $n$-bit string representation is lexicographically the largest).
\end{definition}
Figure~\ref{fig:1} shows the {\em local metadata structure} $G_3$ for the case $n=6$ with mandatory static one shift. Here, the node $v={\tt 101010}$ can be obtained by a static one shift from both the nodes {\tt 110010} and {\tt 101100}. However, $v$ is the mandatory static one shift of only the node containing {\tt 110010}, \emph{and not that of} {\tt 101100}.

\subsection{Construction of the {\em metadata structure} $G$ with Incremental One Shift}

In order to define the complete {\em metadata structure} $G$ for our present scenario, for each subset size $i \in [1 \ldots n - 1], \mid R \mid$ = $n$, two consecutive {\em local metadata structures} $G_i$ and $G_{i+1}$ are connected using the concept of ``Incremental One Shift". Here a directed edge goes from a node $v(S) \in V (G_{\mid S \mid})$ to a node $v(S') \in V(G_{\mid S \mid+1})$. Note that the bit vector representation of the node $v(S')$ has one extra `1' than $v(S)$. 

Figure~\ref{fig:2} shows the {\em metadata structure} $G$ for the case $n=4$. Here, a directed edge goes from a node $v={\tt 0100}$ of $G_1$ to nodes $v_1={\tt 1100}$, $v_2={\tt 0110}$, and $v_3={\tt 0101}$ of $G_2$, where all $v1$, $v_2$ and $v_3$ are the ``Incremental One Shifts" of $v$.

\begin{figure}[t]
\begin{center}
\scalebox{.45}{\includegraphics{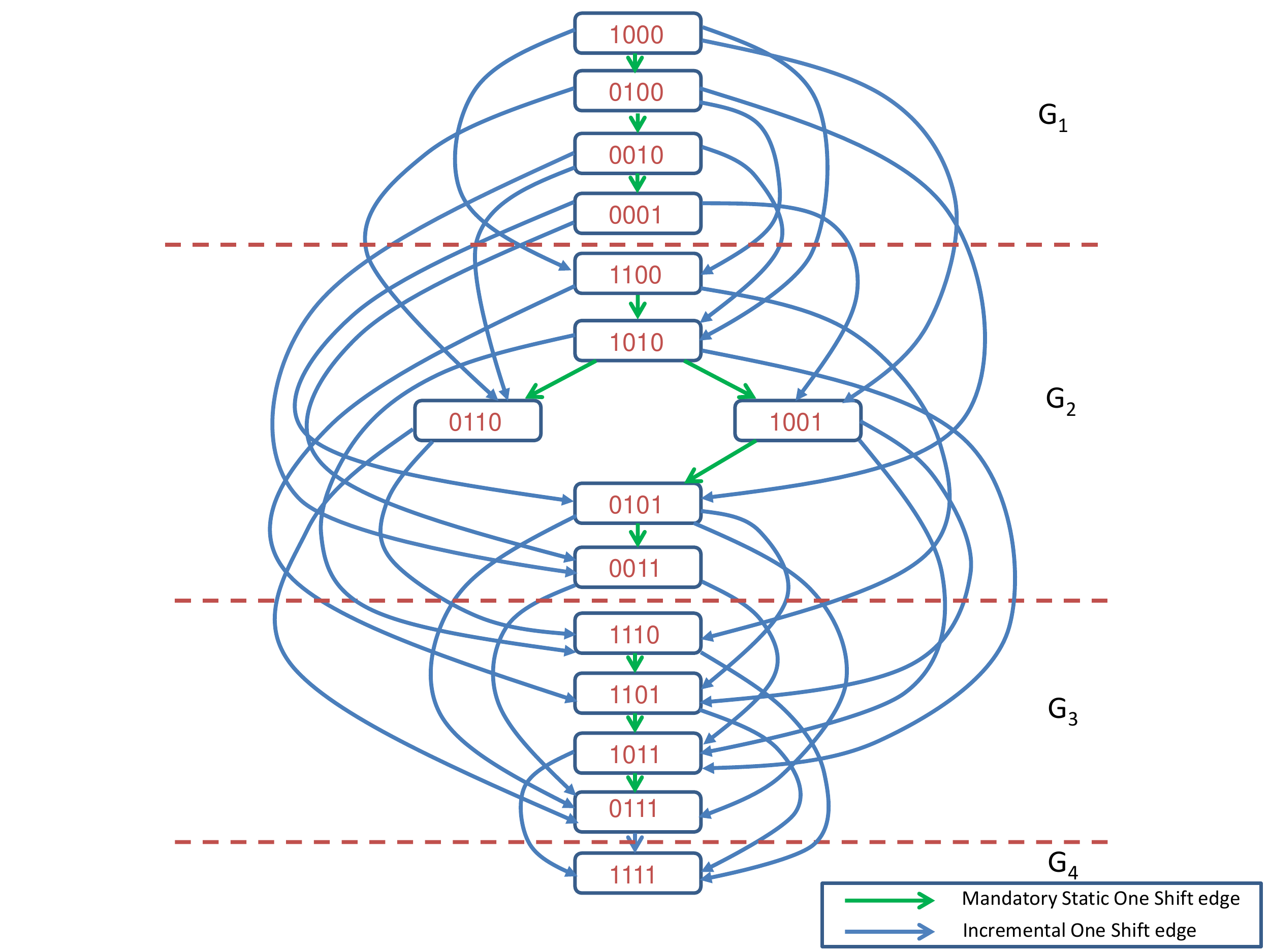}}
\end{center}
\caption{The model {\em metadata structure} $G$ for $n=4$ with \emph{incremental one shift}}
\label{fig:2}
\end{figure}

The definition of incremental one shift leads to the following two observations.
\begin{observation}
\label{lemma:diff}
Let $(p_1, p_2, \ldots, p_{\mid S \mid})$ denote the list of sorted positions of the ${\mid S \mid}$ numbers in a subset $S \subseteq R$ and further let $(p'_1, p'_2, \ldots, p'_{\mid S' \mid})$ denote the list of sorted positions of the ${\mid S' \mid}$ numbers in a subset $S' \subseteq R$, where $\mid S' \mid$ = $\mid S \mid$ $+$ $1$.
Then, $S'$ is an incremental one shift of $S$ if and only if for some $j$, $1 \leq j \le {\mid S' \mid}$, the position index $p'_j$ is not in $(p_1, p_2, \ldots, p_{\mid S \mid})$ and $\sum_i p'_i - 
\sum_i p_i = p'_j$. 
\end{observation}
\begin{observation}
Each node $v(S), S\subseteq R$ of the {\em metadata structure} $G$ has $n$ - $\mid S \mid$ incremental one shift children. 
\end{observation}

\subsection{Query answering with heap}
\label{queryans}

Note that in $G$, two subsets $S_i$ and $S_j$ are comparable if there is a directed path between the two corresponding nodes $v(S_i)$ and $v(S_j)$. However, if there is no path in $G$ between the nodes $v(S_i)$ and $v(S_j)$, then it is required to calculate the values of the summation function $F(S_i)$ and $F(S_j)$ explicitly to find out which one ranks higher in the output list. To facilitate this process, a min-heap $H$ is maintained to store the candidate subsets $S$ according to their key values $F(S)$. Initially, the root node $T$ of the {\em metadata structure} $G$ is inserted in min-heap $H$, with key value $F(T)$. Then, to report the desired top-$k$ subsets, at each step, the minimum element $Z$ of $H$ is extracted and report it as an answer, and then insert each of its children $X$ from $G$ into $H$ with key value $F(X)$. Obviously, the above set of steps have to be performed $k-1$ times until all the top-$k$ subsets are reported.

However, the problem that we face here is the high out degree of each node $v$ in $G$. Here each node can have at most two \emph{static one shift} edges~\cite{smhg_tcs20} but has a high number of \emph{incremental one shift} edges. As a consequence, many nodes have multiple parents in $G$. 
Note that, for any node $v(S)$ of $G$ with multiple parents, we need to insert the subset $S$ or rather the node $v(S)$ into the heap $H$, right after reporting the subset stored in any one of its parent nodes, i.e. when for the first time we extract any of its parents say $u$ from $H$. On the other hand, $v(S)$ can be extracted from $H$ only after all the subsets stored in its parent nodes have been reported as part of the desired result, i.e., all their corresponding nodes have been extracted from $H$. So during the entire lifetime of $v$ in $H$, whenever the subset stored in some other parent of $v$ is reported, either the subset $S$ needs to be inserted again in $H$ or a checking is to be performed whether any node corresponding to $S$ is already there in $H$. The former strategy will lead to duplication in the heap $H$ and the latter one will lead to too much overhead as we then have to then check for prior existence in $H$, for each and every child of any node that will get extracted from $H$. Either way, the time complexity will rise.

To avoid this problem of duplication, whenever a node in $H$ is inserted, we also store the label of that node in a Skip List or a height-balanced binary search tree (AVL tree) $T$, and only insert $v$ to $H$ if $v$ is not already present in $T$. 
The above step has to be performed exactly $k-1$ times till all top-$k$ subsets are generated as output. For a summary of this discussion given above, see Algorithm~\ref{algo1}, which is somewhat similar in principle to the query algorithm that works along with a preprocessing step, presented in our earlier work~\cite{smhg_tcs20}. However, the on-demand version (Algorithm~\ref{algo2}) presented in this work is totally different from that of our earlier work.

\begin{algorithm}

\caption{Top-$k$\_Subsets\_With\_Metadata\_Structure($R[1 \ldots n]$, $G$, $k$)}
\label{algo1}
\begin{algorithmic}[1]
\State Create an empty min-heap $H$;
\State Create an empty binary search tree $T$;
\State Sort the $n$ real numbers of $R$ in non-decreasing order;
\State $Root \leftarrow$ the root node of $G$ (Root node of $G_1$);
\State Insert $Root$ into $H$ with key value $F(Root)$;
\State Insert $Root$ into $T$;
\For{$q \leftarrow 1$ to $k$}
\State $Z \leftarrow$ extract-min($H$); 
\State Output $Z$ as the $q^{th}$ best subset;
\State Delete $Z$ from $T$;
\For{each child $X$ of $Z$ in $G$}
\If{$X$ is not found in $T$}
\State Insert $X$ into $H$ with key value $F(X)$; $\newline$ 
\Comment $F$ is the summation function, such that $F(X) = \sum_{r \in X} r$, for any subset $X\subseteq R$.
\State Insert $X$ into $T$;
\EndIf
\EndFor
\EndFor
\end{algorithmic}
\end{algorithm}

Actually a \emph{min-max-heap} can be used instead of a min-heap, so as to limit the number of candidates in $H$ to be at most $k$. Alternatively, a max-heap $M$ along with $H$ can be used, to achieve the same feat in the following way. Whenever an element is inserted in $H$, it is also inserted in $M$ and an invariant is maintained such that the size of $M$ is always less than or equal to $k$. If it tries to cross $k$, the maximum element from $M$ as well as $H$ are removed, since such an element can never come in the list of top-$k$ elements being the maximum within $k$ elements. Note that Algorithm~\ref{algo1} can be made even more output-sensitive by dynamically limiting the number of elements in the two heaps by $(k - y)$ if $y$ is the number of subsets already reported. This is also being reflected in the pseudo-code of Algorithm~\ref{algo2} later. This leads to the following lemma.
\begin{lemma}
\label{lemma_space_pre}
The extra working space of Algorithm~\ref{algo1}, in addition to that for maintaining $G$, is $O(kn)$.
\end{lemma}
\begin{proof}
By using a min-max-heap, the size of the heap $H$ never exceeds $k$, and so does the size of the AVL tree $T$. Also, each node contains a bit-pattern of size $n$. 
\end{proof}
\begin{lemma}
Apart from the time to sort $R$, Algorithm~\ref{algo1} runs in $O (nk\log_2 k)$ time. 
\end{lemma}
\begin{proof}
Any node $v(S)$ in $G$ has at most $(n$ - $\mid S \mid$ $+$ $2)$ children, and the value of $F(X)$ for all children $X$ can be computed in a total of $O(n$ - $\mid S \mid$ $+$ $2)$ time (using dynamic programming), since there are only $O(1)$ differences between $v$ and any of its children. Now, if the subset reported at the $i^{th}$ step is $S_i$, then $ \sum_{i=1}^{k} n$ - $\mid S_i \mid$ $+$ $2$, i.e. $O(nk)$ insertions and extract-min operations are performed on the min-max-heap $H$, and at most $O(nk)$ search, insertions, and deletions are performed on the AVL tree $T$. Since the size of $H$ and $T$ are bounded from above by $k$, the overall running time is $O(nk+nk\log_2 k) = O(nk\log_2 k)$.
\end{proof}

\begin{figure}[t]
\begin{center}
\scalebox{.45}{\includegraphics{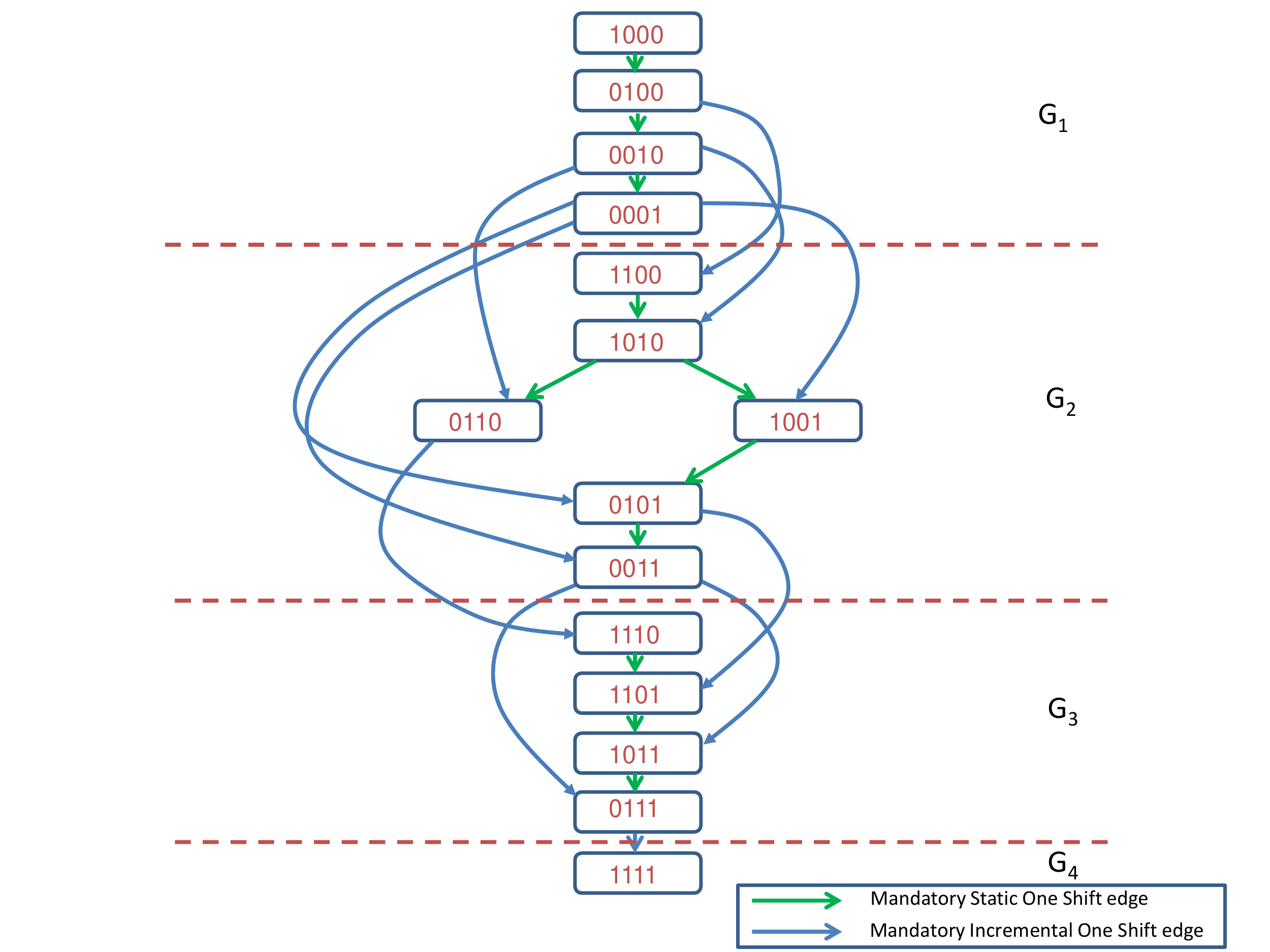}}
\end{center}
\caption{The model {\em metadata structure} $G$ for $n=4$ with \emph{mandatory incremental one shift}}
\label{fig:3}
\end{figure}

\subsection{Modified {\em metadata structure} $G$ -- version I}

In order to reduce the overall time complexity, a natural choice would be to reduce the number of \emph{incremental one shift} edges between two consecutive {\em local metadata structures} $G_i$ and $G_{i+1}$ ($i \in [1 \ldots n - 1]$), so that duplication problem in the heap will automatically get reduced, we define below a mandatory version of the incremental one shift.

\begin{definition}[Mandatory Incremental One Shift]
$S'$ is said to be a \emph{Mandatory Incremental One Shift} of $S$, if (i) $S'$ is an incremental one shift of $S$, and (ii) among all the subsets of $R$, for which $S'$ is an incremental one-shift, $S$ is the one whose position sequence is lexicographically the largest (equivalently, the $n$-bit string representation is lexicographically the smallest).
\end{definition}

The constructed {\em metadata structure} $G$ after applying \emph{mandatory incremental one shift}, exhibits the following interesting properties :

\begin{enumerate}
\item Each valid subset can be reached from the root.
\item The root of $G_1$ has no parent. Every other node of $G_1$ has a unique parent.
\item Each of the roots of the other data structures $G_2$ to $G_n$ has a unique parent (by mandatory incremental one shift) and all other nodes have exactly two parents (one by mandatory incremental one shift and the other from mandatory static one shift).
\item The bit-pattern corresponding to every child of a node can be deduced from the bit-pattern corresponding to that node.
\end{enumerate}

Figure~\ref{fig:3} gives an example of the {\em metadata structure} $G$ for the case $n=4$. Note that $v = {\tt 1101}$ is a mandatory incremental one shift of {\tt 0101}, but
$v$ is \emph{not} a mandatory incremental one shift of {\tt 1100} or {\tt 1001}.

The definition of mandatory incremental one shift directly leads to the following two observations.
\begin{observation}
\label{obs:diff}
Let $(p_1, p_2, \ldots, p_{\mid S \mid})$ denote the list of sorted positions of the numbers in a subset $S \subseteq R$ and
let $(p'_1, p'_2, \ldots, p'_{\mid S' \mid})$ denote the list of sorted positions of the numbers in another subset $S' \subseteq R$, where $\mid S' \mid$ = $\mid S \mid$ $+$ $1$.
Then, $S'$ is a mandatory incremental one shift of $S$ if and only if for some $j$, $1 \leq j \le {\mid S' \mid}$--\\ 
i) the position index $p'_j$ is not in $(p_1, p_2, \ldots, p_{\mid S \mid})$,\\ 
ii) $p'_j<p_1$, and\\ 
iii) $\sum_i p'_i - \sum_i p_i = p'_j$. 
\end{observation}
\begin{observation}
If $(p_1, p_2, \ldots, p_{\mid S \mid})$ denotes the list of sorted positions of the numbers in a subset $S \subseteq R$, then the node $v(S)$ of the {\em metadata structure} $G$ has exactly $(p_1-1)$ mandatory incremental one shift children. 
\end{observation}

The $2^{nd}$ property from Observation~\ref{obs:diff} can be easily verified from Figure~\ref{fig:3}. The nodes containing the subsets {\tt 1010} and {\tt 0110} are the children of the node containing the subset {\tt 0010} which means the subsets {\tt 1010} and {\tt 0110} can be obtained from the subset {\tt 0010} by mandatory incremental one shift. This in fact leads us to the next lemma.

The rationale behind refining the definition of shift in steps is to make the graph $G$ more sparse without disturbing the inherent topological ordering, since the complexity of the algorithm directly depends on the number of edges that $G$ contains. So a last enhancement is further made on the graph $G$ by defining another type of shift below.

\subsection{Modified {\em metadata structure} $G$ -- version II}
Note that many nodes in the DAG $G$ still have high out degrees due to multiple \emph{mandatory incremental one shift} edges. Specially, the bottom most node in each $G_i$ (except $G_n$) has $n-i$ \emph{mandatory incremental one shift} edges. We can remove most of these incremental edges to decrease the number of edges in the DAG and hence its complexity by keeping at most one outgoing incremental edge from each node by redefining the definition of \emph{mandatory incremental one shift}.

\begin{definition}[Modified Mandatory Incremental One Shift]
$S'$ is said to be a \emph{Modified Mandatory Incremental One Shift} of $S$, if (i) $S'$ is a mandatory incremental one shift of $S$, and (ii) among all those subsets of $R$, which are mandatory incremental one shifts of $S$, $S'$ is the one whose position sequence is lexicographically the smallest (equivalently, the $n$-bit string representation is lexicographically the largest).
\end{definition}

\begin{figure}[t]
\begin{center}
\scalebox{.45}{\includegraphics{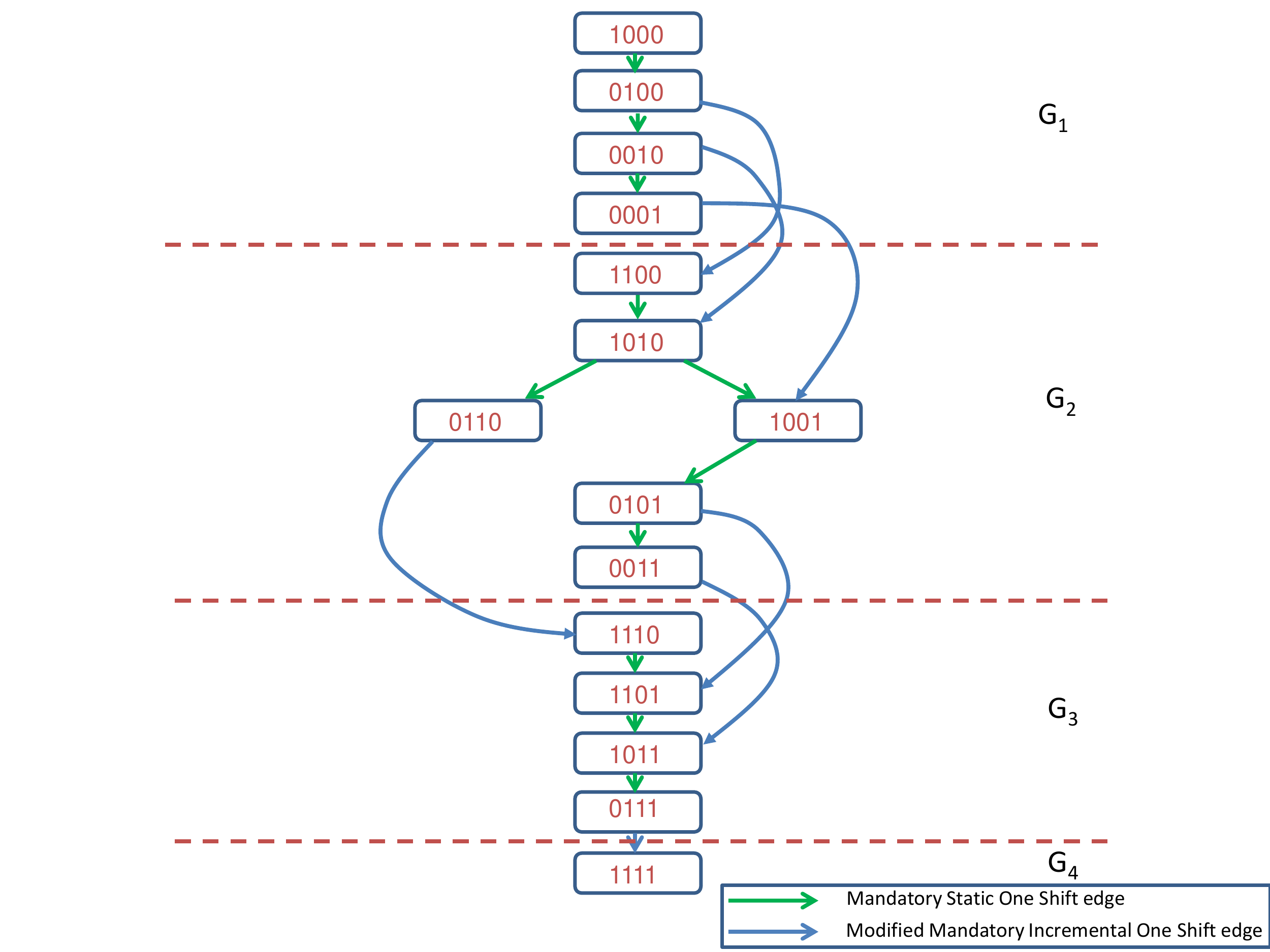}}
\end{center}
\caption{The model {\em metadata structure} $G$ for $n=4$ with \emph{modified mandatory incremental one shift}}
\label{fig:4}
\end{figure}

The {\em metadata structure} $G$, we have, after applying \emph{modified mandatory incremental one shift}, has the following properties:
\begin{enumerate}
\item Each valid subset can be reached from the root.
\item A node $v$ of $G$, now has at most one outgoing edge of incremental type and in total it has at most two children (one by modified mandatory incremental one shift and another by static one shift).
\item The root of $G_1$ has no parent. Every other node of $G_1$ has a unique parent. The roots of the other data structures $G_2$ to $G_n$ has exactly one parent (by modified mandatory incremental one shift) and all the other nodes have at most two parents (by modified mandatory incremental one shift and static one shift).
\item The bit-pattern corresponding to every child of a node can be deduced from the bit-pattern corresponding to that node.
\end{enumerate}

Consider Figure~\ref{fig:4} for the modified DAG $G$ for $n = 4$. Note that the subsets {\tt 1001}, {\tt 0101}, and {\tt 0011} are all \emph{mandatory incremental one shifts} of {\tt 0001} but according to the definition, only {\tt 1001} is the \emph{modified mandatory incremental one shift} of {\tt 0001}.

The definition of modified mandatory incremental one shift easily leads to the following observation.
\begin{observation}
If $(p_1, p_2, \ldots, p_{\mid S \mid})$ denotes the list of sorted positions of the ${\mid S \mid}$ numbers in a subset $S \subseteq R$, then the node $v(S)$ of the {\em metadata structure} $G$ has only one modified mandatory incremental one shift children, where $p_1>1$. 
\end{observation}

Clearly, after the introduction of \emph{modified mandatory incremental one shift}, any node $v \in V(G)$ has now at most one outgoing incremental edge. However, having multiple outgoing edges is not the only issue that affects the run-time complexity, having multiple incoming edges also does so.

Let us consider a node $v(S_c)$ in $G$ that has two parents $v(S_{p1})$ and $v(S_{p2})$, where $S_c$ is the modified mandatory incremental one shift of $S_{p1}$ and also the mandatory static one shift of $S_{p2}$ respectively and obviously, $S_c, S_{p1}, S_{p2} \subseteq R$. Clearly, both $S_{p1}$ and $S_{p2}$ will occupy higher ranks than $S_c$ in the desired top-$k$ result. So, the subset $S_c$ can never be reported prior to $S_{p1}$ and $S_{p2}$. Also note that the subset $S_c$ is inserted in the heap just after one of its parents is reported (as well as deleted) from the heap as part of the desired result. It will stay in the heap at least till its other parent is reported (as well as deleted) from the heap. The problem of duplication arises exactly when the second reporting happens and calls for a reinsertion of the subset $S_c$ again in the heap. So, at that time, either it is required to reinsert it or execute a routine to check whether $S_c$ is already there in the heap and hence, the current implementation also suffers from the problem of `checking for node duplication'.

For example, consider Figure~\ref{fig:3}, where {\tt 1001} has two parents {\tt 0001} (by modified mandatory incremental one shift) and {\tt 1010} (by static one shift). Without loss of generality, if the parent {\tt 0001} is reported first as part of desired top-$k$ subsets then as its child node, {\tt 1001} will be inserted into the heap and it will stay in heap till its next parent {\tt 1010} is reported. It happens so, as parents are always ranked higher than the child in the DAG $G$. But when the parent {\tt 1010} is reported, as its child node, {\tt 1001} is again supposed to be added in the heap following the reporting logic using the heap. So if added without checking, it would have caused multiple insertion of the same node thereby causing unnecessary increase in runtime. The other option is to check for node duplication before inserting the child node, which would also cause an increase in the run-time complexity. 
However, it is obvious that the problem of node duplication can be removed altogether if every node in $G$ has only one parent.

\subsection{The final structure of $G$}

To summarize, let us recollect that the complete {\em metadata structure} $G$ is a collection of $n$ {\em local metadata structures} $G_1$ to $G_n$ where each {\em local metadata structure} $G_i$ is capable of reporting all the top-$k$ subsets of size $i$, efficiently. Each node of $G_i$ is reachable from its root and it is possible to use a heap $H_i$ to report these subsets of size $i$. Now, in order to improve the time complexity, it is needed to remove the node duplication problem altogether, i.e., we want that each node $v$ of $G$ to have only one parent. However, it must also be ensured that each valid subset can be reached from the root of $G$ (which is also the root of $G_1$) and also the subsets corresponding to all the children of a node $v$ can be easily deduced from the subset corresponding to $v$.
\begin{figure}
\begin{center}
\scalebox{.45}{\includegraphics{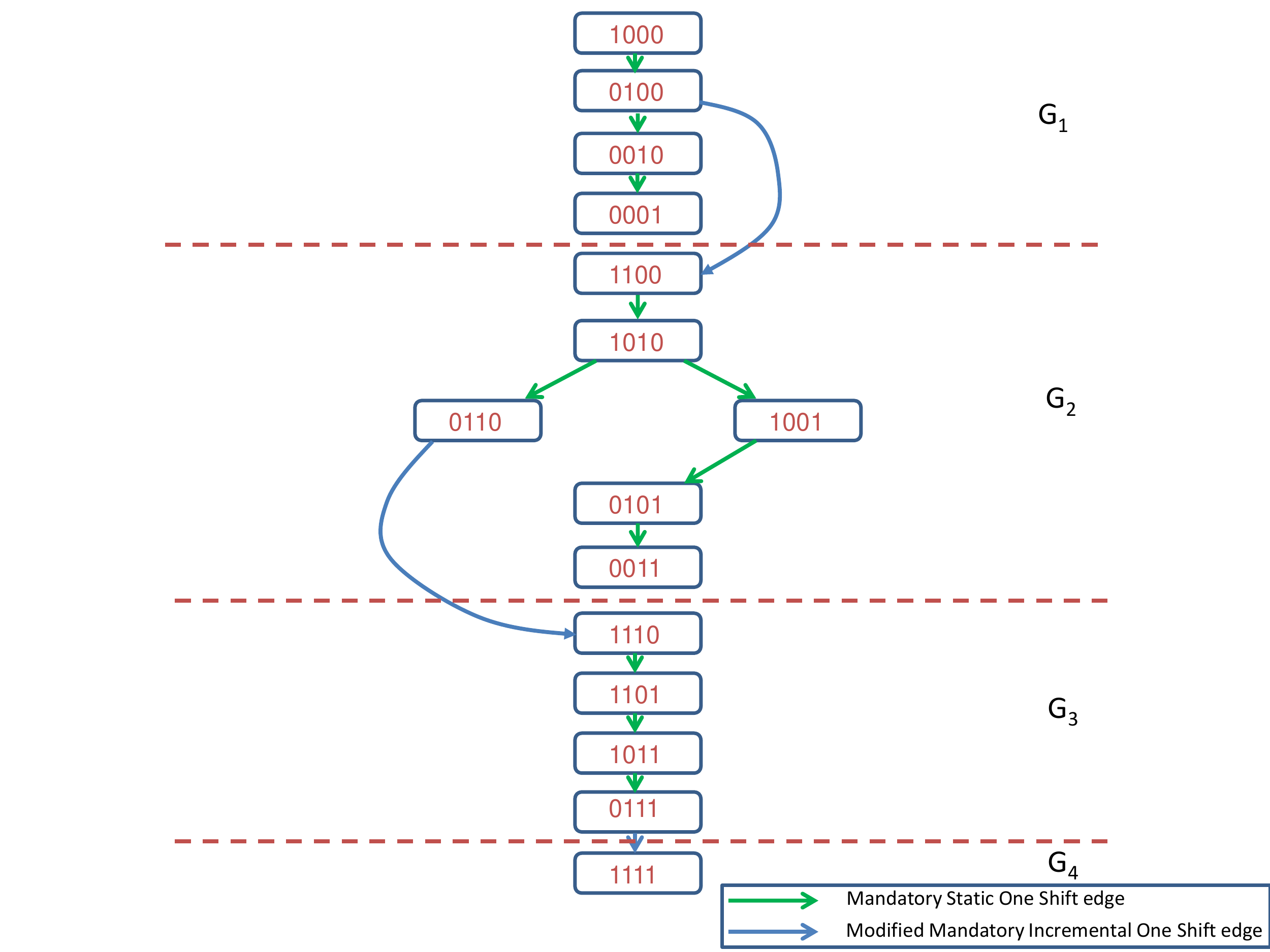}}
\end{center}
\caption{The model {\em metadata structure} $G$ for $n=4$ with one parent for each node}
\label{fig:5}
\end{figure}

In the current design of $G$, the root of $G_1$ has no parent. Other nodes of $G_1$ has one parent. The roots of other data structures $G_2$ to $G_n$ has one parent (by modified mandatory incremental one shift) and all other nodes have at least one parent and at most two parents (one by static one shift and the other possibly by modified mandatory incremental one shift). The desired goal of having every node (other than the root of $G$) with only one parent can be achieved, simply by omitting all incremental edges from $G$ except the ones that lead to the roots of $G_i$ for each subset size $i \in \{2, \ldots, n \}, \mid R \mid = n$. Deleting these edges will reduce the graph in Figure~\ref{fig:4} to that of Figure~\ref{fig:5}, which gives us the final DAG $G$ for $n=4$. 

The {\em metadata structure} $G$, we have, after applying these last set of modifications, has the following properties:

\begin{enumerate}
\item Each valid subset can be reached from the root.
\item Every node of $G$ other than the root has only one parent. The root nodes of $G_2$ to $G_n$ are connected to their respective parents by a modified mandatory incremental one-shift edge and all other nodes are connected to their corresponding parents by a mandatory static one-shift edge.
\item The bit-pattern corresponding to every child of a node can be deduced from the bit-pattern corresponding to the node in constant time (see Lemma~\ref{lemma_o(1)}).
\end{enumerate}

\begin{lemma}
\label{lemma_two_child}
Every node of this final DAG $G$ has at most two children.
\end{lemma}
\begin{proof}
At most two edges of type mandatory static one-shift can go out any node~\cite{smhg_tcs20}. Also, it is our claim that any node from where a modified incremental one-shift edge can originate can have at most one mandatory static child. This is because such a node has a bit-pattern that has the 1st bit as '0' immediately followed by only one contiguous block b of '1's (say a subset {\tt 0111000}, where $n = 7$). Only in one case, where the length of this block b is just 1 (the node being {\tt 0100000}, $n = 7$), we can have a mandatory static child ({\tt 00100000}). However, for all other cases, the nodes (like {\tt 0110000}, {\tt 0111000}, {\tt 0111100} and so on) can not have any mandatory static child in G~\cite{smhg_tcs20}. For example, {\tt 01110000} cannot have {\tt 01101000} as its child, since its parent is {\tt 10101000} by the definition of mandatory static one-shift.

Hence the proof.
\end{proof}

\subsection{Generation of top-$k$ subsets by traversal of implicit DAG $G$ on demand}
\label{G:ondemand}

Note that instead of creating this final {\em metadata structure} $G$ as a part of preprocessing, only its required portions can be constructed on demand, more importantly the edges of this DAG $G$ are not even needed to be explicitly connected. The real fact is portion of the graph that is explored just remains implicit among the nodes that are inserted into $H$. Here, we create and evaluate a subset corresponding to a node only if its parent node $u$ in the implicit DAG gets extracted. It saves considerable storage space as well as improves the run-time complexity. See Algorithm~\ref{algo2} for a detailed pseudo-code showing how the whole algorithm can be implemented. Lemma~\ref{lemma_o(1)} in turn establishes the fact that the algorithm really performs as intended.

\begin{algorithm}

\caption{Top-$k$\_Subsets\_On\_Demand($R[1 \ldots n]$, $k$)$\newline$ 
Structure of a node of the implicit DAG: bit-pattern $B[1 \ldots n]$, subsetSize: $ls$, aggregation-Value: $F$, tuple of three array indices $(p_1, p_2, p_3)$}
\label{algo2}
\begin{algorithmic}[.9]
\State Create an empty min-heap $H$ and an empty max-heap $M$;
\State Sort the $n$ real numbers of $R$ in non-decreasing order;$\newline$
\Comment{Create the root node $Root$ of $G$}
\State $Root.B[1] \leftarrow1$;
\For{$i \leftarrow 2$ to $\mid R \mid$}
\State $Root.B[i] \leftarrow 0$;
\EndFor
\State $Root.(p_1, p_2, p_3) \leftarrow$ $(0, 1, 1)$;
\State $Root.ls \leftarrow1$;
\State $Root.F \leftarrow R[1]$;
\State insert Root in $H$ as well as $M$; 
\State $count \leftarrow 1$;
\For{$q \leftarrow 1$ to $k$}
\State $currentNode \leftarrow$ extractMin($H$); 
\State $count \leftarrow count$ $-$ $1$;
\State \parbox[t]{\dimexpr\linewidth-\algorithmicindent}{Output the subset for bit-pattern $currentNode.B[1 \ldots n]$ and also its value $currentNode.F$ as the $q^{th}$ best subset; \strut}$\newline$ 
\Comment{Get one mandatory static one-shift child (if any)}
\If{$(childSType1 \leftarrow \Call{genManStaticType1}{currentNode})$}
\State insert $childSType1$ in $H$ as well as $M$;
\State $count \leftarrow count$ + $1$;
\EndIf
$\newline$
\Comment{Get the other mandatory static one-shift child (if any)}
\If{$(childSType2 \leftarrow \Call{genManStaticType2}{currentNode})$}
\State insert $childSType2$ in $H$ as well as $M$;
\State $count \leftarrow count$ + $1$;
\EndIf
$\newline$
\Comment{Get mandatory incremental static one-shift child (if any)}
\If{$(childMI \leftarrow \Call{genManIncremental}{currentNode})$}
\State insert $childMI$ in $H$ as well as $M$;
\State $count \leftarrow count$ + $1$;
\EndIf
$\newline$
\Comment{Remove the maximum element from both heaps if required, as we need to output only $k$ subsets; also removing more than one node at a time is never required since any parent node can have at most 2 children}
\If{$count > k - q$}
\State $extraNode \leftarrow extractMax(M)$;
\State remove $extraNode$ also from $H$;
\State $count \leftarrow count$ $-$ $1$;
\EndIf
\EndFor
\end{algorithmic}
\end{algorithm}

\begin{algorithm}
\begin{algorithmic}[1]
\Function{genManStaticType1}{node ParentNode}
\State node $childNode \leftarrow NULL$;
\If{$ParentNode.p_1 > 1$ and $ParentNode.p_1 < n$}

\If{$ParentNode.B[ParentNode.p_1$ $+$ $1] = 0$}
\State $childNode \leftarrow ParentNode$; \Comment{copy ParentNode into childNode}
\State $childNode.B[childNode.p_1] \leftarrow 0$;
\State $childNode.B[childNode.p_1$ $+$ $1] \leftarrow 1$;
\State $childNode. p_1 \leftarrow childNode. p_1$ $+$ $1$;
\If{$childNode.p_3 = childNode.p_1$ $-$ $1$}
\State $childNode.p_3 \leftarrow childNode.p_3$ $+$ $1$;
\EndIf
\State $childNode.F \leftarrow$ $childNode.F - R[childNode.p_1$ $-$ $1]$;
\State $childNode.F \leftarrow$ $childNode.F$ $+$ $R[childNode.p_1]$;

\EndIf
\EndIf
\State \Return $childNode$;
\EndFunction
%\newline
%//Get other mandatory static one-shift child (if any)

%\algstore{myalg}
\end{algorithmic}
\end{algorithm}

\begin{algorithm}
\begin{algorithmic}[1]
\Function{genManStaticType2}{node ParentNode}
\State node $childNode \leftarrow NULL$;
\If{$ParentNode.p_2 > 0$ and $ParentNode.p_2 < n$}
\State $childNode \leftarrow ParentNode$; \Comment{copy ParentNode into childNode}
\State \textbf{Swap}($childNode.B[childNode.p_2], childNode.B[childNode.p_2 + 1]$);
\State $childNode.p_1 \leftarrow childNode.p_2$ $+$ $1$; \Comment{$p_2$ now points to leftmost $0$}
\If{$childNode.p_3 = childNode.p_2$}
\State $childNode.p_3 \leftarrow childNode.p_2$ $+$ $1$;
\EndIf
\State $childNode.p_2 \leftarrow childNode.p_2$ $-$ $1$;
\State $childNode.F \leftarrow$ $childNode.F$ $-$ $R[childNode.p_2$ $+$ $1]$;
\State $childNode.F \leftarrow$ $childNode.F$ $+$ $R[childNode.p_2$ $+$ $2]$;

\EndIf
\State \Return $childNode$;
\EndFunction
%\newline
%//Get other mandatory static one-shift child (if any)

%\algstore{myalg}
\end{algorithmic}
\end{algorithm}

\begin{algorithm}
\begin{algorithmic}[1]
\Function{genManIncremental}{node ParentNode}
\State node $childNode \leftarrow NULL$;
\If{$ParentNode.p_1=2$ and $ParentNode.p_2 =0$ and $\newline$ \phantom{xx} $ParentNode.p_3= ParentNode.ls$ $+$ $1$}
\State $childNode \leftarrow ParentNode$; \Comment{copy ParentNode into childNode}
\State $childNode.B[1] \leftarrow 1$;

\State $childNode.p_1 \leftarrow 0$;
\State $childNode.p_2 \leftarrow childNode.p_3$;
\State $childNode.ls \leftarrow childNode.ls$ $+$ $1$;
\State $childNode.F \leftarrow$ $childNode.F$ $+$ $R[1]$;
\EndIf
\State \Return $childNode$;
\EndFunction
%\newline
%//Get other mandatory static one-shift child (if any)

%\algstore{myalg}
\end{algorithmic}
\end{algorithm}

Note that apart from the root of the {\em metadata structure}, each node can be created in $O(1)$ incremental time. Since we create the subset corresponding to node $v$ only when its parent $u$ in the implicit DAG is extracted, and there is only a difference of $O(1)$ bits between $u$ and $v$, the creation as well as evaluation of the subset for $v$ can be done in $O(1)$ incremental time, if the subset of $u$ is known.

\begin{lemma}
\label{lemma_o(1)}
Given any bit-pattern of a subset and its corresponding value, the bit-pattern of the subset and the corresponding value for any of its child can be generated in $O(1)$ incremental time. Also the type of edges that come out of its corresponding node can be determined in $O(1)$ time.
\end{lemma}
\begin{proof}
Note that in the implicit DAG $G$, the bit-pattern $B[1 \ldots n]$, stored in any node $u$ represents one non-empty subset $S\subseteq R$. Depending on the bit-pattern stored in the node $u$, we create the bit-pattern of its child node $v$ as well as evaluate its value. Moreover, from this bit-pattern we can also determine the type of edge that would implicitly connect $u$ and $v$. 

Note that $v$ is a mandatory static child of $u$ only if one of the following holds~\cite{smhg_tcs20}:

\begin{enumerate}
\item The bit-pattern of $v$ can be generated from the bit-pattern of $u$ by swapping the {\tt 1} that immediately appears after the first sequence of {\tt 0s} in the bit-pattern of $u$; moreover, the right neighbour of this {\tt 1} must be a {\tt 0} for the swapping to be valid. It does not matter whether the bit-pattern stored in $u$ starts with {\tt 0} or {\tt 1}. For example, if $u$ has a bit-pattern like {\tt 00010011000} then $v$ will have the pattern {\tt 00001011000} or if $u$ has the pattern {\tt 1100101101} then $v$ has the pattern {\tt 1100011101}.

\item The bit-pattern of $u$ starts with {\tt 1}. Then the bit-pattern of $v$ can be generated from the bit-pattern of $u$ by swapping the rightmost {\tt 1} appearing in the first contiguous block of {\tt 1s} with its neighbouring {\tt 0}. For example, the following two bit patterns {\tt 11100011} and {\tt 11010011} can be possible candidates for the nodes $u$ and $v$ respectively.
\end{enumerate}

Moreover, it is also easy to verify that a node $u$ can have a modified mandatory incremental child of $u$ if and only if the bit-pattern of $u$ starts with single $0$, immediately followed by only one contiguous block b of {\tt 1s}, and the bit-pattern of $v$ can be generated from $u$ by just replacing this first {\tt 0} by {\tt 1}.

In order to generate the bit-patterns of the children efficiently, in any node $u$ corresponding to the subset $S$, a tuple of three pointers $(p_1, p_2, p_3)$ is maintained, and also subset size $\mid S \mid$ in addition to its bit-pattern and the aggregation value. The first pointer $p_1$ points to the position of first {\tt 1} that immediately appears after the first block of {\tt 0(s)} in $u$, the pointer $p_2$ points to the position of the rightmost {\tt 1} appearing in the first contiguous block of {\tt 1s} in case the leftmost bit of $u$ is a {\tt 1}, otherwise it remains {\tt 0} and the pointer $p_3$ points to the rightmost {\tt 1} in the bit-pattern of $u$. Note that $p_3$ can never be {\tt 0}. It is always in $[1 \ldots n$]. Also, note that $p_1$ can never take the value {\tt 1}, it can vary between $[2 \ldots n]$.

For the root node of $G$ (the root node of local DAG $G_1$), it is easy to verify that $p_1=0$ (i.e. not pointing to any position of the bit-pattern), $p_2$ = $\mid S \mid$, and $p_3$ = $\mid S \mid$, i.e., $p_2=1$, and $p_3=1$, since for this node clearly the size of the subset $S$ is {\tt 1}. Note that for the root node, the bit-pattern starts with a single {\tt 1} followed by $(n -1)$ {\tt 0s} and setting up the bit-pattern requires $O(\mid R \mid) = O(n)$ time. The three pointers and the value (which is basically the least element in the set $R$) can be set in $O(1)$ additional time.

It may be recalled here that any portion of the DAG $G$ is never created explicitly, rather whenever we extract and report a node $u$ from the heap, we create each of its child nodes but do not connect it to $u$ with an edge to form the DAG explicitly, and also evaluate its value so that it can be inserted in the heap $H$ for further manipulation.

Now, after extracting a node $u$ from $H$, the following steps are required to be performed :

\begin{enumerate}
\item Check if the value of $p_1$ is in $[2 .. n]$ and additionally if $B[p_1 + 1]=0$. If so, then from $u$ we get a mandatory static child node $v$, just by swapping $B[p_1]$ and $B[p_1+1]$. Now, we just have to adjust the pointers of $v$. Set $p_1$(v) = $p_1$(u) + 1. $p_2$ will remain unchanged, i.e., $p_2(v)$ = $p_2(u)$. If in $u$, $p_3 = p_1$ , then also set $p_3(v) = p_3(u) + 1$, otherwise $p_3$ also remains unchanged.

\item Check the pointer $p_2$. If it is in $[1 .. (n$ $-$ $1)]$, then $B[p_2]$ and $B[p_2+1]$ are swapped in order to create the other mandatory static child node $v$ of $u$. Set $p_2(v)$ = $p_2(u) - 1$. If $p_3(u) = p_2(u)$, then set $p_3(v) = p_2(u) + 1$. Finally, set $p_1(v) = p_2(u) + 1$.

\item Check the pointer triplet $(p_1, p_2, p_3)$ for a specific list of values $(2, 0, \mid S\mid$ + 1) and if yes, set $B[1]=1$ in the bit-pattern of $u$ keeping all the other bits unchanged, to have the modified mandatory incremental child node $v$ of $u$. Here, set $p_1(v) = 0$, $p_2(v) = p_3(v) = p_3(u)$ and also increment subset size $\mid S \mid$ of $v$ by {\tt 1}.
\end{enumerate}

From the above discussion it is clear that given any bit-pattern of a subset stored in a node $u$, the bit-pattern of the subset for any of its child stored in the node $v$ can be generated in $O(1)$ incremental time. Also the type of edge that comes out of its corresponding node can be determined in $O(1)$ time.

Moreover, evaluation of the value of the bit-pattern stored in $v$ can be done from the value of the bit-pattern stored in $u$ in $O(1)$ time by just one addition and at most one subtraction.

Hence the proof. 
\end{proof} 

Here also, a \emph{min-max-heap} or a max-heap can be used in addition to the min-heap, so as to limit the number of candidates in $H$ to be at most $k$. This gives the following two lemmas.
\newline
\begin{lemma}
\label{lemma_space_pre}
The working space required for Algorithm~\ref{algo2} is $O(kn)$.
\end{lemma}
\begin{proof}
By using a min-max-heap, the size of the heap $H$ can be restricted to $k$. Even without using it, the number of insertions possible in $H$ is bounded by $1 + 2(k - 1) = 2k - 1$ from above, since every extraction from $H$ can cause at most two insertions except the last one. Since each node contains $n$ bits, the total space required is $O(kn)$.
\end{proof}

\begin{lemma}
\label{algo2:time}
Apart from the time required to sort $R$, Algorithm~\ref{algo2} runs in $O (n+k\log_2 k)$ time. If it is required to report the subsets, it will take $O(nk + k\log_2 k)$ time. 
\end{lemma}
\begin{proof}
It takes $O(n)$ time to create the root node. Then each node $v$ in $G$ has at most two children, and the bit-pattern as well as the value $F(X)$ for a child $X$ can be computed in $O(1)$ time from the bit-pattern and the value of its parent. Also, at most $2k-1$ insertions and $k$ extract-min operations are performed on the min-max-heap $H$. Since the size of $H$ is bounded by $k$, the overall running time is $O(n+(2k - 1)\log_2 k + k\log_2 k) = O(n+k\log_2 k)$. However, if we have to report the elements of the subsets, then after each of the $k$ extractions, 
it is required to decode the bit-string, which would take $O(n)$ time. Hence, the total time required will be $O(nk + k\log_2 k)$.
\end{proof}

\subsection{Getting rid of the bit string}

Note that in lemma~\ref{lemma_o(1)}, it is already established that the bit string of a child node varies from its parent only at a constant number of places and also the subset sum of any child node can be generated from its parent node in $O(1)$ time, since it requires only a constant number additions and subtractions with the sum stored in the parent node.
Actually, with a slight modification, Algorithm~\ref{algo2} can be run even without the explicit storage of the bit-pattern at each node. Here it is required to maintain another pointer value that points to the position of the first {\tt 1} after the position pointed by pointer $p_1$. It can now be noted that the types of edges that come out from any node in $G$ can be determined in $O(1)$ time. Once the indices of the elements present in the subset corresponding to the parent node are known, it is possible to generate the array indices of the elements and from them the actual element(s) of each subset can be found using $O(1)$ incremental time for each node. In the following lemma it is shown that the running time improves to $O(k \log_2 k)$ and hence becomes independent of $n$. It also gets exhibited in the next section from the results of our implementation, where we see that the runtime falls drastically. Recall that the subset sum of any child node can also be generated from its parent node in $O(1)$ time. 
\begin{lemma}
\label{algo_sum}
The modified version of algorithm~\ref{algo2} that works without storing the bit stream runs in $O(k \log_2 k)$ time.
\end{lemma}
\begin{proof}
As mentioned above, in our solution now, we are not maintaining the bit-pattern at each node, rather, we are using just one extra pointer value that points to the position of first {\tt 1} after the position pointed by pointer $p_1$. The size of the root node is $O(1)$ as it contains a single element or rather the index of the first element of the array containing the elements in ascending order. From then on, for each node that gets generated, the time required is $O(1)$. Then from the proof of Lemma~\ref{algo2:time} it is obvious that the total time required by the modified algorithm~\ref{algo2} will be $O(k + k \log_2 k)$ = $O(k \log_2 k)$.
\end{proof}
\section{Experimental results}
\label{sec:result}
Algorithm~\ref{algo2} was implemented in C and was tested with varying values of $n$ and $k$. The experiments were conducted on a desktop powered with an Intel Xeon 2.4 GHz quad-core CPU and 32GB RAM. The operating system loaded in the machine was Fedora LINUX version 3.3.4.
\begin{table}[H]

\begin{center}
\scalebox{.51}{ \begin{tabular}{|>{\centering}p{.6 cm}|>{\centering}p{1.5 cm}|>{\centering}p{1.8 cm}|>{\centering}p{1.7 cm}|>{\centering}p{1.5 cm}|>{\centering}p{1.8 cm}|>{\centering}p{1.7 cm}|>{\centering}p{1.5 cm}|>{\centering}p{1.8 cm}|>{\centering}p{1.7 cm}|>{\centering}p{1.5 cm}|>{\centering}p{1.8 cm}|>{\centering}p{1.7 cm}|c|}
\hline
& \multicolumn{3}{c|}{k=1000} & \multicolumn{3}{c|}{k=10000} & \multicolumn{3}{c|}{k=100000} & \multicolumn{3}{c|}{k=1000000} & k=10000000\\ 
\hline
n & \tabincell{c}{Time with \\Bit Vector \\$(B)$} & \tabincell{c}{Time without \\Bit Vector\\$(V)$} & \tabincell{c}{Speed Up \\$(S= B /\ V)$}& \tabincell{c}{Time with \\Bit Vector \\$(B)$} & \tabincell{c}{Time without \\Bit Vector\\$(V)$} & \tabincell{c}{Speed Up \\$(S= B /\ V)$}& \tabincell{c}{Time with \\Bit Vector \\$(B)$} & \tabincell{c}{Time without \\Bit Vector\\$(V)$} & \tabincell{c}{Speed Up \\$(S= B /\ V)$}& \tabincell{c}{Time with \\Bit Vector \\$(B)$} & \tabincell{c}{Time without \\Bit Vector\\$(V)$} & \tabincell{c}{Speed Up \\$(S= B /\ V)$} & \tabincell{c}{Time without \\Bit Vector $(V)$} \\
\hline
100 & 0.0069 & 0.0020 & $3.45 X$ & 0.0541 & 0.0196& $2.76 X$& 0.4062 & 0.1846 & $2.20 X$ & 4.4800 & 2.2010 & $2.04 X$ & 25.8228\\
\hline
200 & 0.0105 & 0.0026 & $4.04 X$& 0.0706 & 0.0292 & $2.42 X$ & 0.5893 & 0.2013 & $2.93 X$ & 6.3452 & 2.1755 & $2.92 X$ & 24.9263 \\
\hline
300 & 0.0142 & 0.0024 & $5.92 X$ & 0.1733 & 0.0283 & $6.12 X$& 0.7775 & 0.2001 & $3.89 X$ & 8.2391& 2.1625 & $3.81 X$ & 24.4716\\
\hline
400 & 0.0173 & 0.0021 & $8.24 X$& 0.1857 & 0.0293 & $6.34 X$ & 0.9568 & 0.2001 & $4.78 X$ & 10.0176 & 2.1052 & $4.76 X$ & 23.7235 \\
\hline
500 & 0.0193 & 0.0019 & $10.16 X$ & 0.1905 & 0.0218 & $8.74 X$ & 1.1319 & 0.1902 & $5.95 X$ & 11.7948& 2.1720 & $5.43 X$ & 24.0756 \\
\hline
600 & 0.0229 & 0.0023 & $9.96 X$ & 0.2199 & 0.0262 & $8.39 X$ & 1.3318 & 0.2039 & $6.53 X$ & 14.2409 & 2.2006 & $6.47 X$ & 33.5502 \\
\hline
700 & 0.0246 & 0.0022 &$11.18 X$ & 0.1853 & 0.0289 & $6.41 X$ & 1.4907 & 0.2037 & $7.32 X$ & 20.6880 & 2.2287 & $9.28 X$ & 33.8687\\
\hline
800 & 0.0281 & 0.0022 & $12.77 X$ & 0.1749 & 0.0304 & $5.75 X$ & 1.6638 & 0.2114 & $7.87 X$ & 30.6208& 2.1917&$13.97 X$ & 32.0629 \\
\hline
900 & 0.0334 & 0.0022 & $15.18 X$ & 0.2852 & 0.0324 & $8.80 X$& 1.8561 & 0.2180 & $8.51 X$ & 35.0497& 2.2686 &$15.45 X$ & 33.5968 \\
\hline
1000 & 0.0364 & 0.0022 &$16.55 X$& 0.3173 & 0.0286 &$11.09 X$ & 2.0802 & 0.2135 & $9.74 X$ & 49.1019 & 2.3635 & $20.76 X$ & 39.2978 \\
\hline
\end{tabular}}
\end{center}
\caption{Comparison of runtimes in seconds for the two variants for varying values of $n$ and $k$}
\label{tab1}
\end{table}

\begin{figure}[H]
\centering % <-- added

\begin{subfigure}{0.5\textwidth}
\includegraphics[scale = .3]{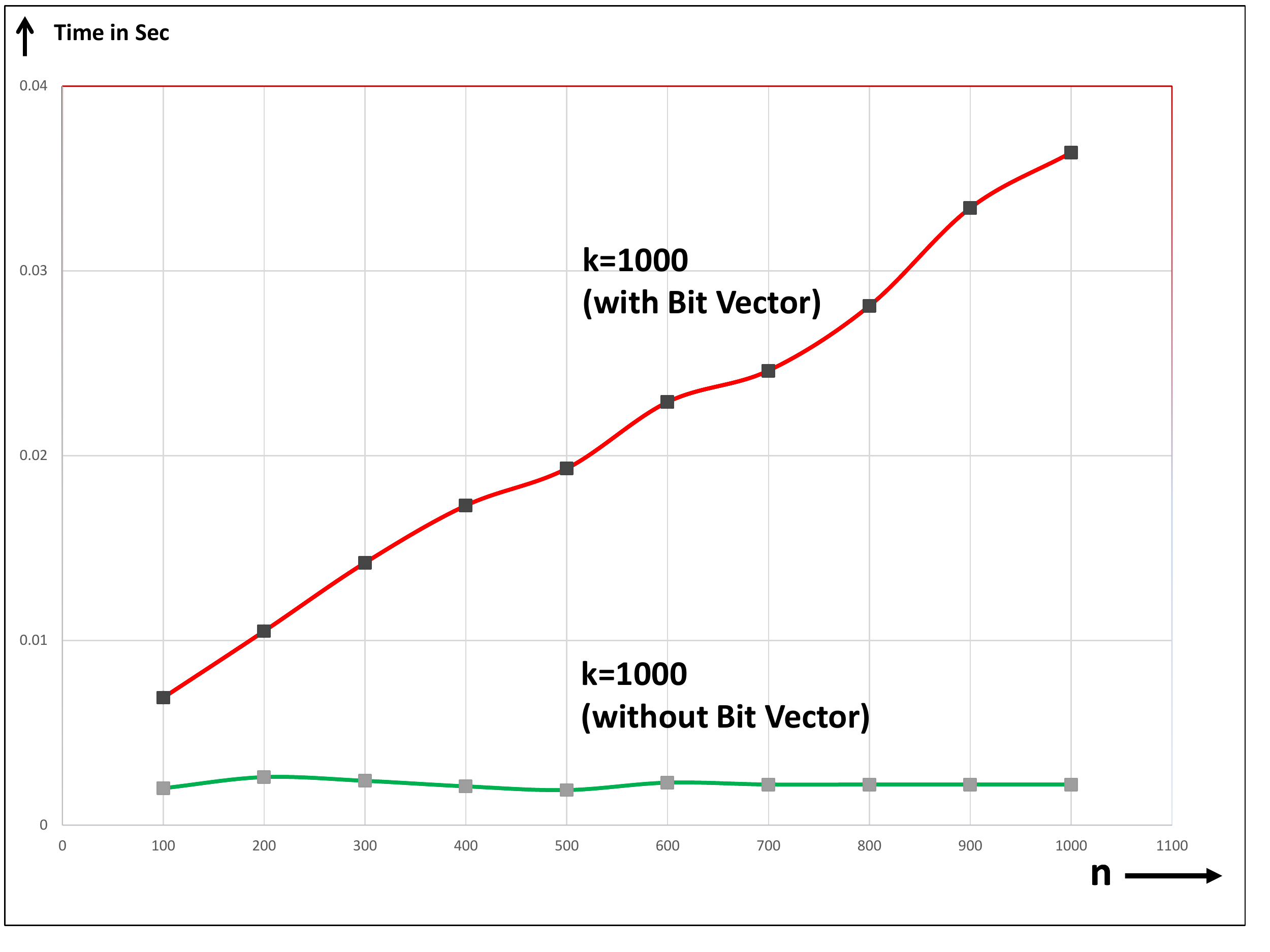}

\caption{Variation of runtime with $n$ for $k=1000$}

\end{subfigure}\hspace*{.4 cm}% <-- added
\begin{subfigure}{0.5\textwidth}
\includegraphics[scale =.3]{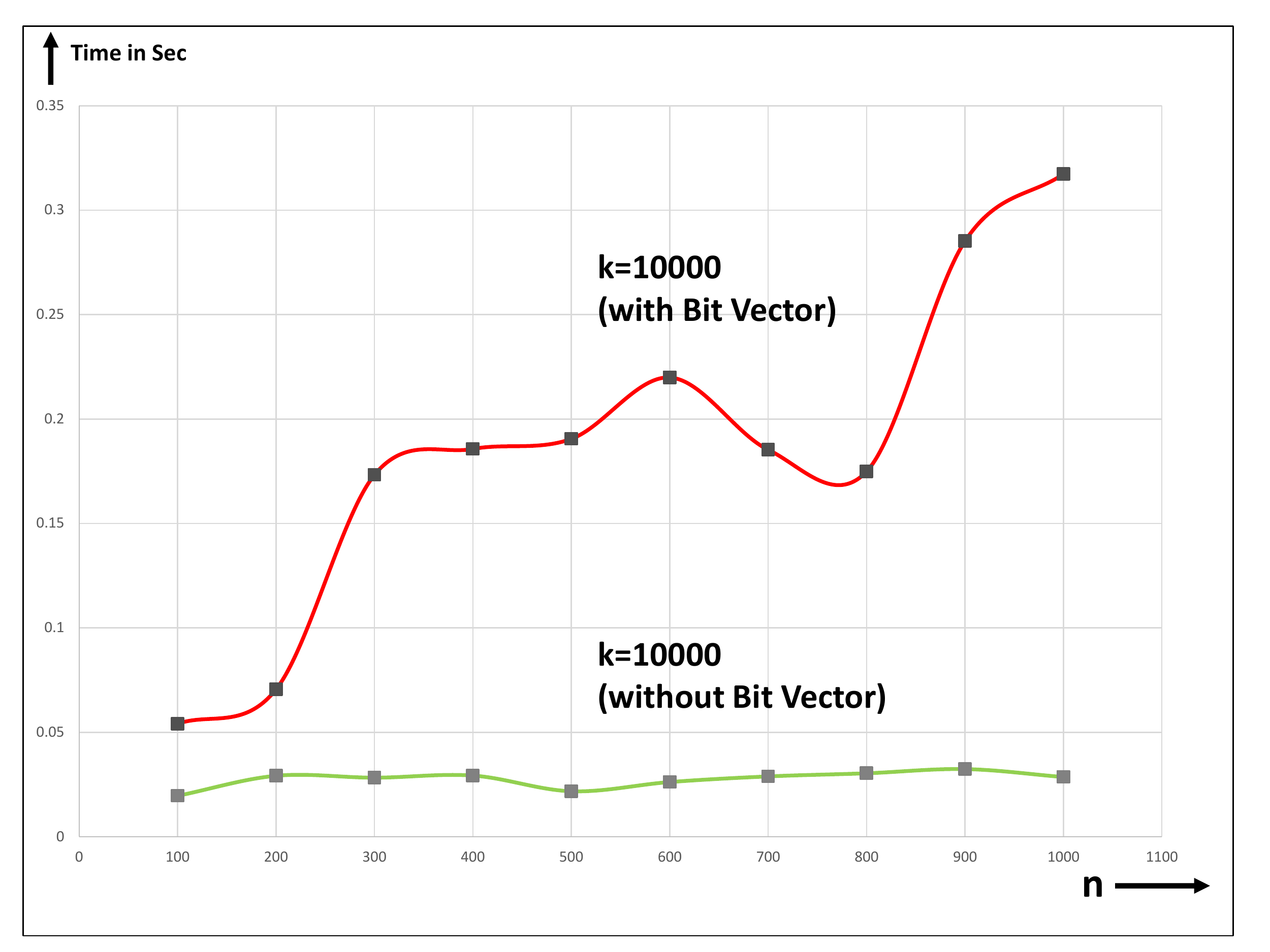}

\caption{Variation of runtime with $n$ for $k=10000$}
\end{subfigure}% <-- added

\medskip
\centering % <-- added

\begin{subfigure}{0.5\textwidth}
\includegraphics[scale = .3]{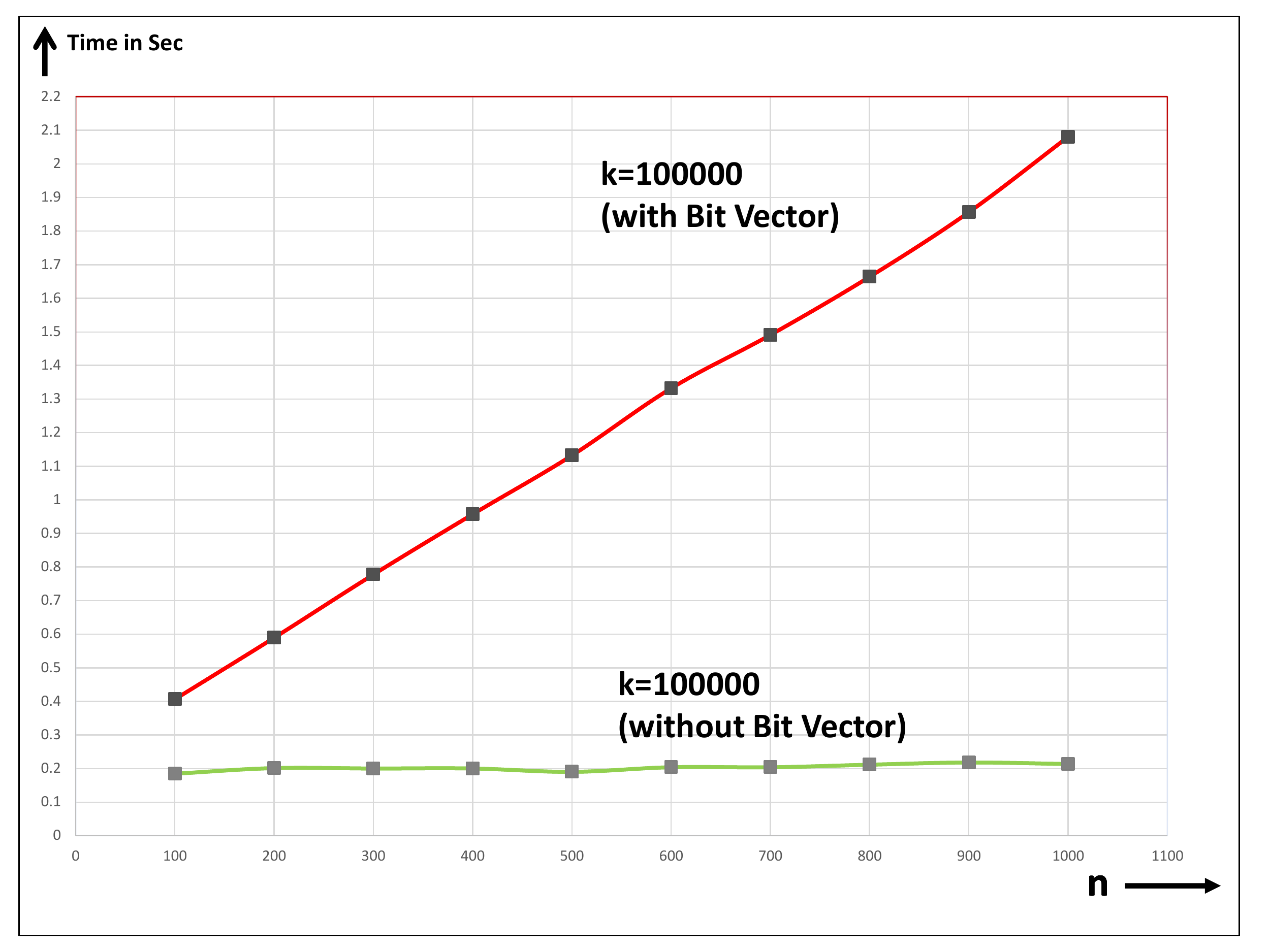}

\caption{Variation of runtime with $n$ for $k=100000$}

\end{subfigure}\hspace*{.3 cm}% <-- added
\begin{subfigure}{0.5\textwidth}
\includegraphics[scale = .3]{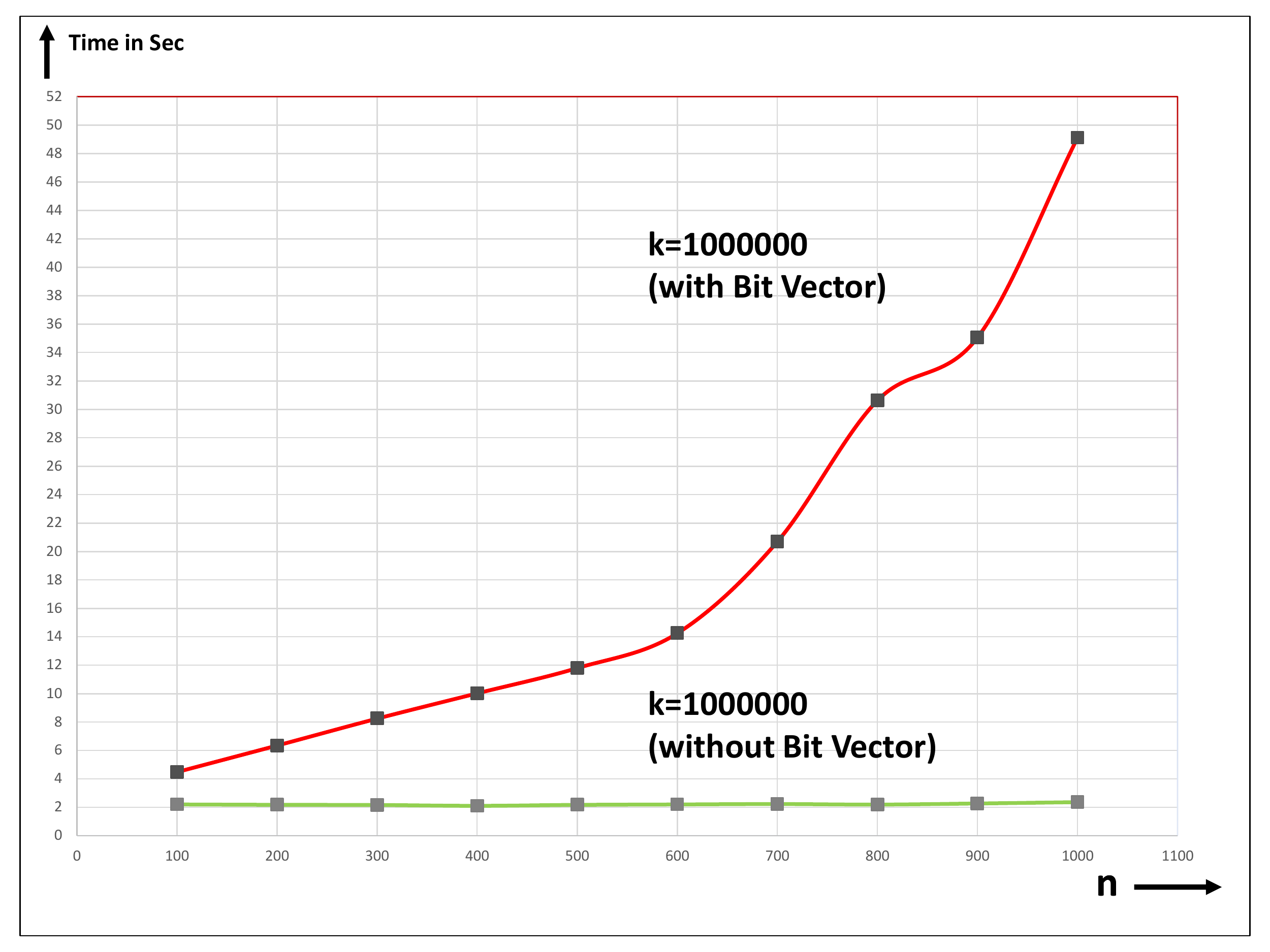}

\caption{Variation of runtime with $n$ for $k=1000000$}
\end{subfigure} 
\caption{Plots of comparing of runtimes in seconds for the two variants with $n$ for varying values of $k$}
\label{fig:6}
\end{figure}
For each specific choice of $(n, k)$, five datasets were generated randomly, and after running each of the algorithms it was found that the runtimes and other reported parameters hardly vary for that specific pair of $n$ and $k$. The results are presented in Table~\ref{tab1} and Table~\ref{tab2} and we made sure that whenever 
two algorithms were compared, they were run on exactly the same dataset.

Table~\ref{tab1} presents the runtimes of the two algorithms, one storing the bit vector in each node explicitly (Algorithm~\ref{algo2}) and the other one being the modified version of the same algorithm where the bit string is not stored and the elements of the subset are generated on the fly by making a few alterations of the constituent elements of the subset corresponding to its parent node in the implicit graph.
\begin{figure}[H]
\centering % <-- added
\hspace*{-1cm}
\begin{subfigure}{0.5\textwidth}
\includegraphics[scale = .35]{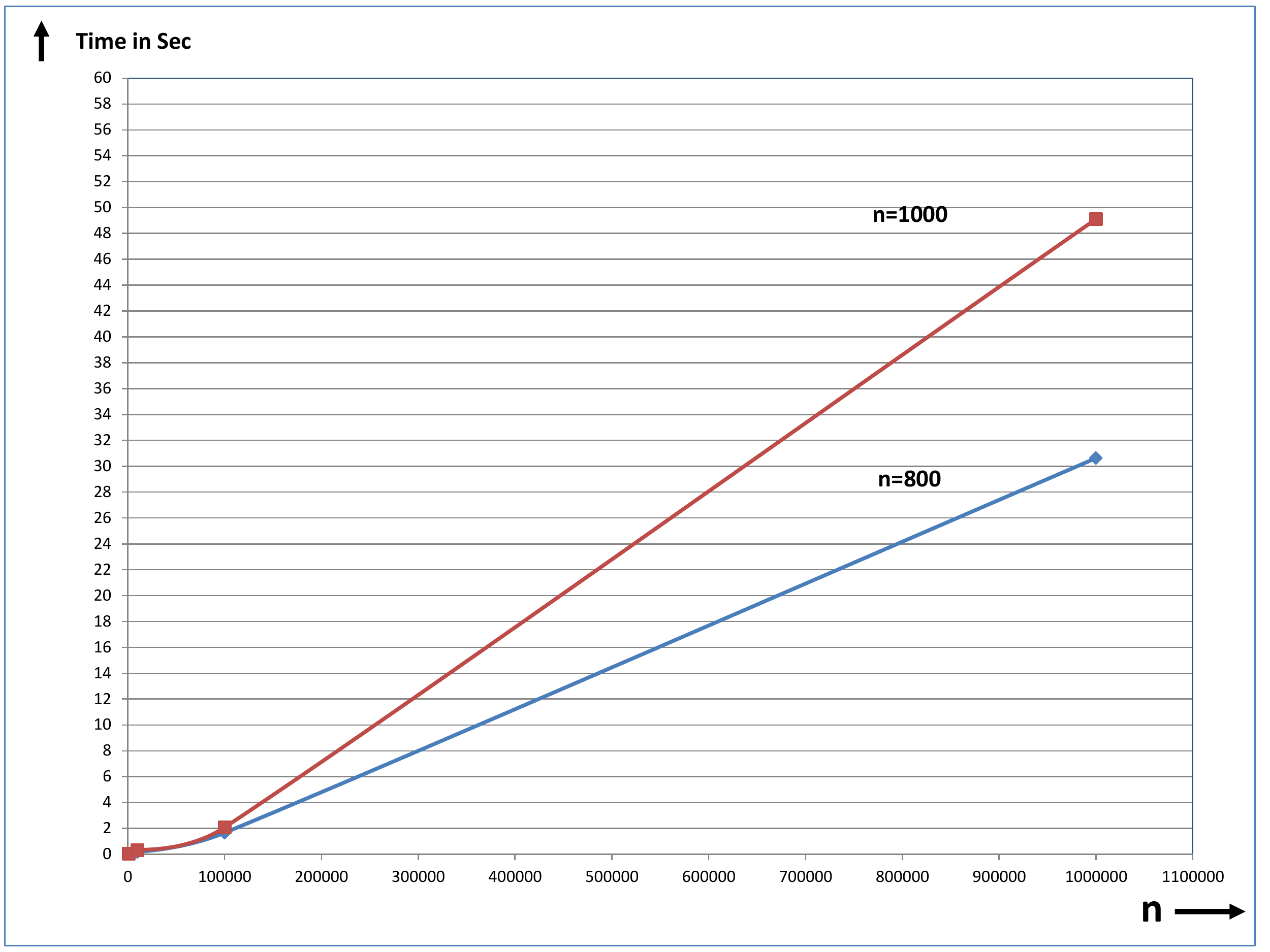}

\caption{Variation of runtime for the with vector variant with $k$ for $n=800$ and $n=1000$}

\end{subfigure}\hspace*{1.2cm}% <-- added
\begin{subfigure}{0.5\textwidth}
\includegraphics[scale =.35]{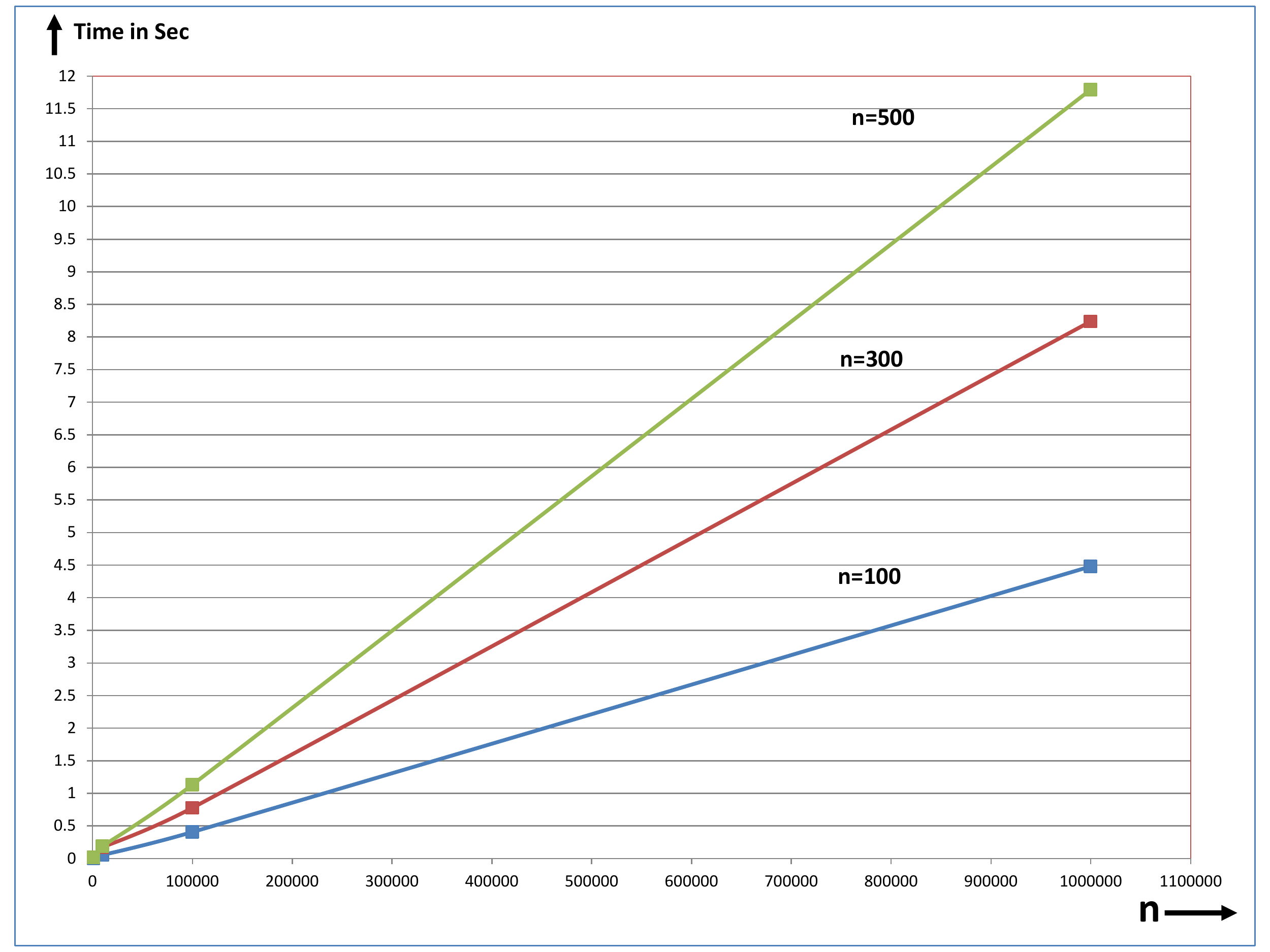}

\caption{Variation of runtime for the with vector variant with $k$ for $n=100, n=300$ and $n=500$}
\end{subfigure}% <-- added

\medskip
\centering % <-- added
\hspace*{-1.5cm}
\begin{subfigure}{0.5\textwidth}
\includegraphics[scale = .4]{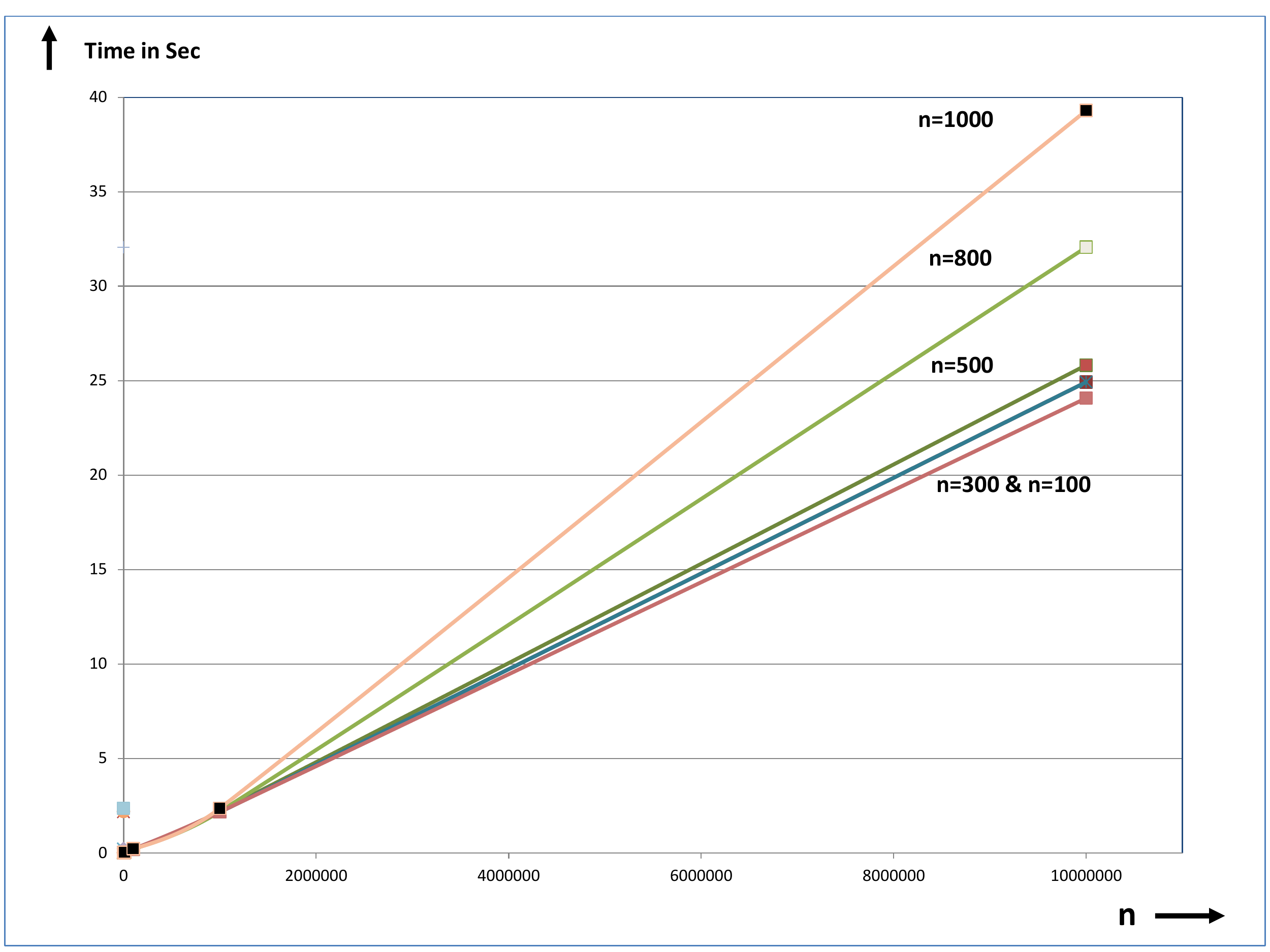}
\vspace*{-.5cm} 
\caption{Variation of runtime for the without vector variant with $k$ for varying values of $n$}

\end{subfigure}% <-- added
\caption{Plots of runtimes in seconds for the two variants with $k$ for varying values of $n$}
\label{fig:7}
\end{figure}

The latter version clearly outperforms the former one and it is found that the speedup is sometimes as high as $20X$ for higher values of $k$. The modified algorithm can even run in less than a minute, where $k$ was as high as $10^7$. 
The former algorithm, however, was getting very slow at those values of $k$ and we did not wait till the runs completed.

\begin{table}

\begin{center}
\scalebox{.6}{ \begin{tabular}{|c|c|c|c|c|c|c|c|c|c|c|}
\hline
& & \multicolumn{2}{c|}{Total time (in sec)} & & \multicolumn{2}{c|}{\tabincell{c}{Total No. of \\Entries in the Heap}} & & \multicolumn{2}{c|}{Peak Size of the Heap} & \\
\hline
n & k& \tabincell{c}{Existing \\Method\\$(B)$} & \tabincell{c}{Our \\Method\\$(V)$} & \tabincell{c}{Speed Up \\$(S= B /\ V)$} & \tabincell{c}{Existing \\Method\\$(P)$} & \tabincell{c}{Our \\Method\\$(Q)$} & \tabincell{c}{\% Improvement\\$(T_I= \frac{(P-Q)}{P} \times 100)$} & \tabincell{c}{Existing \\Method\\$(R)$} & \tabincell{c}{Our \\Method\\$(S)$} & \tabincell{c}{\% Improvement \\$(T_P= \frac{(R-S)}{R} \times 100)$}\\

\hline
100 & 1000 & 0.0009 & 0.0007 & $1.29 X$ & 2001 &1132 & 43.43 &1001 &136 & 86.41 \\
\hline
100 & 10000 & 0.0114 & 0.0091 & $1.25 X$ & 20001 &11078 & 44.61 &10001 &1085 & 89.15 \\
\hline
100 & 100000& 0.1076 & 0.0998 & $1.08 X$& 200001 &109531 & 45.23 &100001 &9531 & 90.47 \\
\hline
100 &1000000 & 1.2933 & 1.0105 & $1.28 X$ & 2000001 &1083508 & 45.82 &1000001 &83519 & 91.65 \\
\hline
100 &10000000 & 12.7919& 11.5919 & $1.10 X$ & 20000001&10745231 & 46.27 &10000001 &745270 & 92.55 \\
\hline

200 & 1000 & 0.0008 &0.00068 & $1.18 X$ &2001 &1102 & 44.93 &1001 &107 & 89.31 \\
\hline
200 & 10000 & 0.0089 & 0.0072 & $1.24 X$ &20001 &10848 & 45.76 &10001 &854 & 91.46 \\
\hline
200 & 100000&0.1068 & 0.0819 & $1.30 X$ &200001 &107377 & 46.31 &100001 &7381 & 92.62 \\
\hline
200 &1000000 & 1.2451 & 0.9553 & $1.30 X$ & 2000001&1063607 & 46.82 &1000001 &63627 & 93.64 \\
\hline
200 &10000000 & 12.1236 & 10.6818 & $1.13 X$ & 20000001&10557415 & 47.21 &10000001 &557428 & 94.43 \\

\hline
300 & 1000 &0.0007 & 0.0006 & $1.17 X$ &2001 &1123& 43.88 &1001 &126 & 87.41 \\
\hline
300 & 10000 & 0.0082 & 0.0066 & $1.24 X$ &20001 &10997& 45.02 &10001 &997 & 90.03 \\
\hline
300 & 100000&0.1073 & 0.0798 & $1.34 X$ &200001 &109175 & 45.41 &100001 &9175 & 90.83 \\
\hline
300 &1000000 & 1.2064 &0.9547 & $1.26 X$ & 2000001&1081802 & 45.91 &1000001 &81802 & 91.82 \\
\hline
300 &10000000 & 12.0638 & 10.7924 & $1.12 X$ & 20000001&10737463 & 46.31 &10000001 &737465 & 92.63 \\

\hline
400 & 1000 &0.0007 &0.0006 & $1.17 X$ &2001 &1069& 46.58 &1001 &69 & 93.11 \\
\hline
400 & 10000 &0.0081 &0.0064 & $1.27 X$ &20001 &10630& 46.85 &10001 &631 & 93.69 \\
\hline
400 & 100000& 0.1065 & 0.0762 & $1.40 X$ &200001 &105278 & 47.36 &100001 &5280 & 94.72 \\
\hline
400 &1000000 & 1.2034 & 0.8952 & $1.34 X$ & 2000001&1047571 & 47.62 &1000001 &47588 & 95.24 \\
\hline
400 &10000000 & 11.7885 & 9.5206 & $1.24 X$ & 20000001&10427513 & 47.86 &10000001 &427588 & 95.72 \\

\hline
500 & 1000 &0.0007 &0.00059&$1.19 X$ &2001 &1076& 46.23 &1001 &77 & 92.31 \\
\hline
500 & 10000 &0.0082 &0.0062 & $1.32 X$ &20001 &10523& 47.39 &10001 &525 & 94.75 \\
\hline
500 & 100000& 0.1068 & 0.0755 & $1.42 X$ &200001 &104460 & 47.77 & 100001 &4460 & 95.54 \\
\hline
500 &1000000 & 1.1767 & 0.9201 & $1.28 X$ & 2000001&1039204 & 48.04 &1000001 &39215 & 96.08 \\
\hline
500 &10000000 & 11.9332 &10.6270 & $1.12 X$ & 20000001&10346223 & 48.27 &10000001 &346254 & 96.54 \\

\hline
600 & 1000 &0.0007 &0.0006 & $1.17 X$ &2001 &1076& 46.23 &1001 &76 & 92.41 \\
\hline
600 & 10000 &0.0097 &0.0088 & $1.10 X$ &20001 &10613& 46.94 &10001 &616 & 93.84 \\
\hline
600 & 100000& 0.1057 &0.0774 & $1.37 X$ &200001 &105353 & 47.32 &100001 &5355 & 94.65 \\
\hline
600 &1000000 & 1.1732 & 0.8996 & $1.30 X$ & 2000001&1048164 & 47.59 &1000001 &48183 & 95.18 \\
\hline
600 &10000000 & 12.0589 &11.0571 & $1.09 X$ & 20000001&10437258 & 47.81&10000001 &437265 & 95.63 \\

\hline
700 & 1000 &0.0007 &0.0006 & $1.17 X$ &2001 &1073& 46.38 &1001 &75 & 92.51 \\
\hline
700 & 10000 &0.0078 &0.0064 & $1.22 X$ &20001 &10740 & 46.30 &10001 &747 & 92.53 \\
\hline
700 & 100000&0.1060 &0.0776 & $1.37 X$ &200001 &106428 & 46.79 &100001 &6436 & 93.56 \\
\hline
700 &1000000 &1.1557 &0.9502 & $1.22 X$ & 2000001&1058889 & 47.06 &1000001 &58892 & 94.11 \\
\hline
700 &10000000 & 12.0418 &11.3002 & $1.07 X$ & 20000001&10543429 & 47.28 &10000001 &543467 & 94.57 \\

\hline
800 & 1000 &0.0007 &0.0006 & $1.17 X$ &2001 &1053& 47.38 &1001 &53 & 94.71 \\
\hline
800 & 10000 &0.0080 &0.0062& $1.29 X$ &20001 &10346& 48.27 &10001 &347 & 96.53 \\
\hline
800 & 100000& 0.1048 &0.0721 & $1.45 X$ &200001 &102680 & 48.66 &100001 &2682 & 97.32 \\
\hline
800 &1000000 &1.1492 &0.8254 & $1.39 X$ & 2000001&1024118 & 48.79 &1000001 &24123 & 97.59 \\
\hline
800 &10000000 & 11.8349 &9.5207 & $1.24 X$ & 20000001&10213912 & 48.93 &10000001 &213916 & 97.86 \\

\hline
900 & 1000 &0.0007 &0.0006 & $1.17 X$ & 2001 &1066& 46.73 &1001 &68 & 93.21 \\
\hline
900 & 10000 &0.0079 &0.0064 & $1.23 X$ &20001 &10643& 46.79 &10001 &645 & 93.55 \\
\hline
900 & 100000&0.1040 &0.0768 & $1.35 X$ &200001 &106278 & 46.86 &100001 &6279 & 93.72 \\
\hline
900 &1000000 &1.1248 &0.8773 & $1.28 X$ & 2000001&1058804 & 47.06 &1000001 &58804 & 94.12 \\
\hline
900 &10000000 & 11.7640 &10.5033 & $1.12 X$ & 20000001&10551227 & 47.24 &10000001 &551264 & 94.49 \\

\hline
1000 & 1000 &0.0007 &0.0006 & $1.17 X$ &2001 &1079& 46.08 &1001 &79 & 92.11 \\
\hline
1000 & 10000 &0.0081 &0.0064 & $1.27 X$ &20001 &10622& 46.89 &10001 &631 & 93.69 \\
\hline
1000 & 100000&0.1052 &0.0766 & $1.37 X$&200001 &105589 & 47.21 &100001 &5591 & 94.41 \\
\hline
1000 &1000000 &1.1618 &0.8831 & $1.32 X$ & 2000001&1051421 & 47.43 &1000001 &51425 & 94.86 \\
\hline
1000 &10000000 & 12.0205 &10.7482 & $1.12 X$ & 20000001&10479269 & 47.60 &10000001 &479286 & 95.21 \\
\hline
\end{tabular}}
\end{center}
\caption{Comparison of our algorithm with an existing algorithm reporting only the sums}
\label{tab2}
\end{table}

\vspace*{-1.6cm}
\begin{figure}[H]
\centering % <-- added
\hspace*{-1.4 cm}
\begin{subfigure}{0.55\textwidth}
\includegraphics[width=\linewidth]{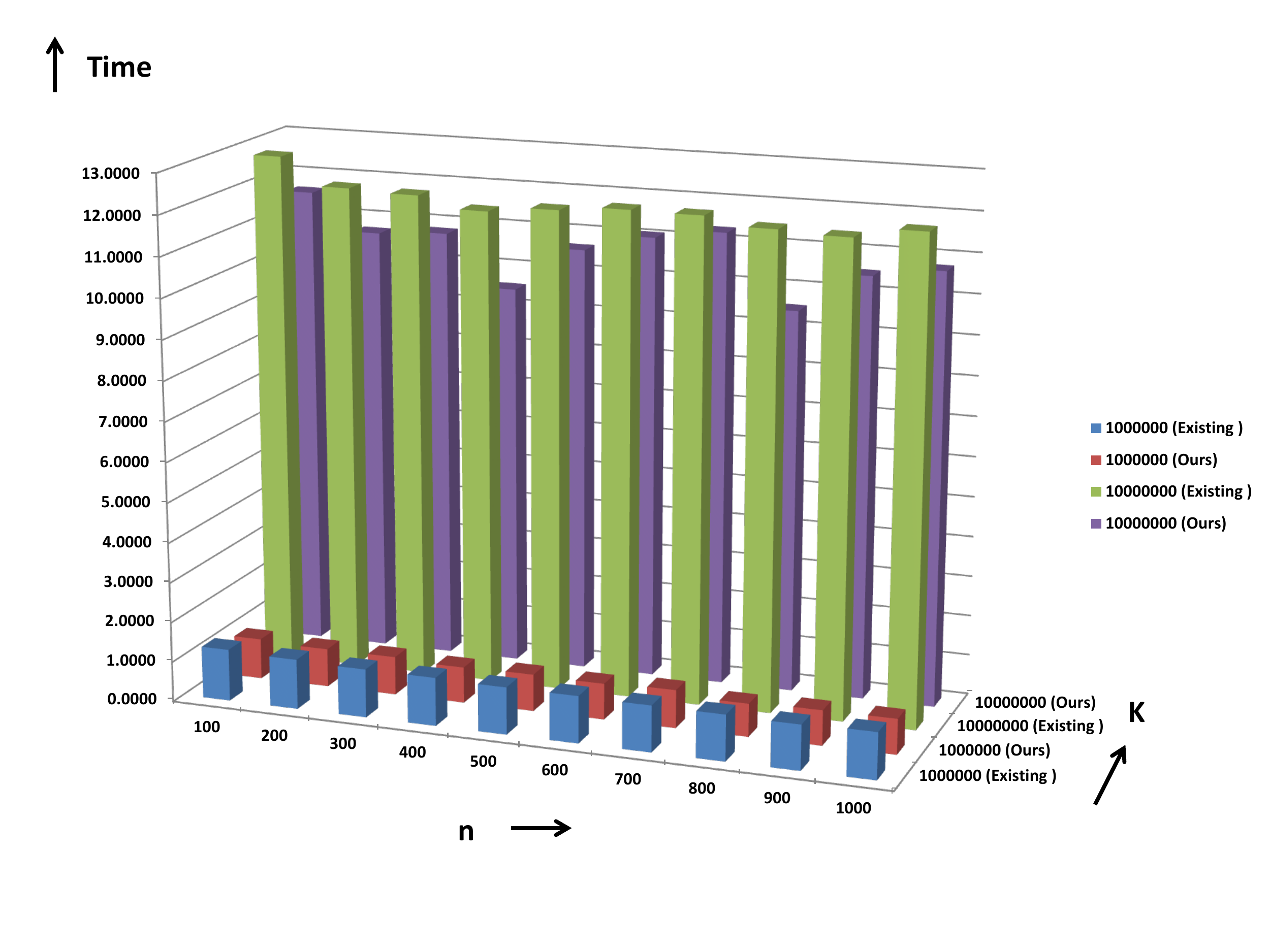}
\vspace*{-1.2cm}
\caption{Variation of runtime for $k=1000000$ and $k=10000000$}
\end{subfigure}\hspace*{.4 cm}% <-- added
\begin{subfigure}{0.55\textwidth}
\includegraphics[width=\linewidth]{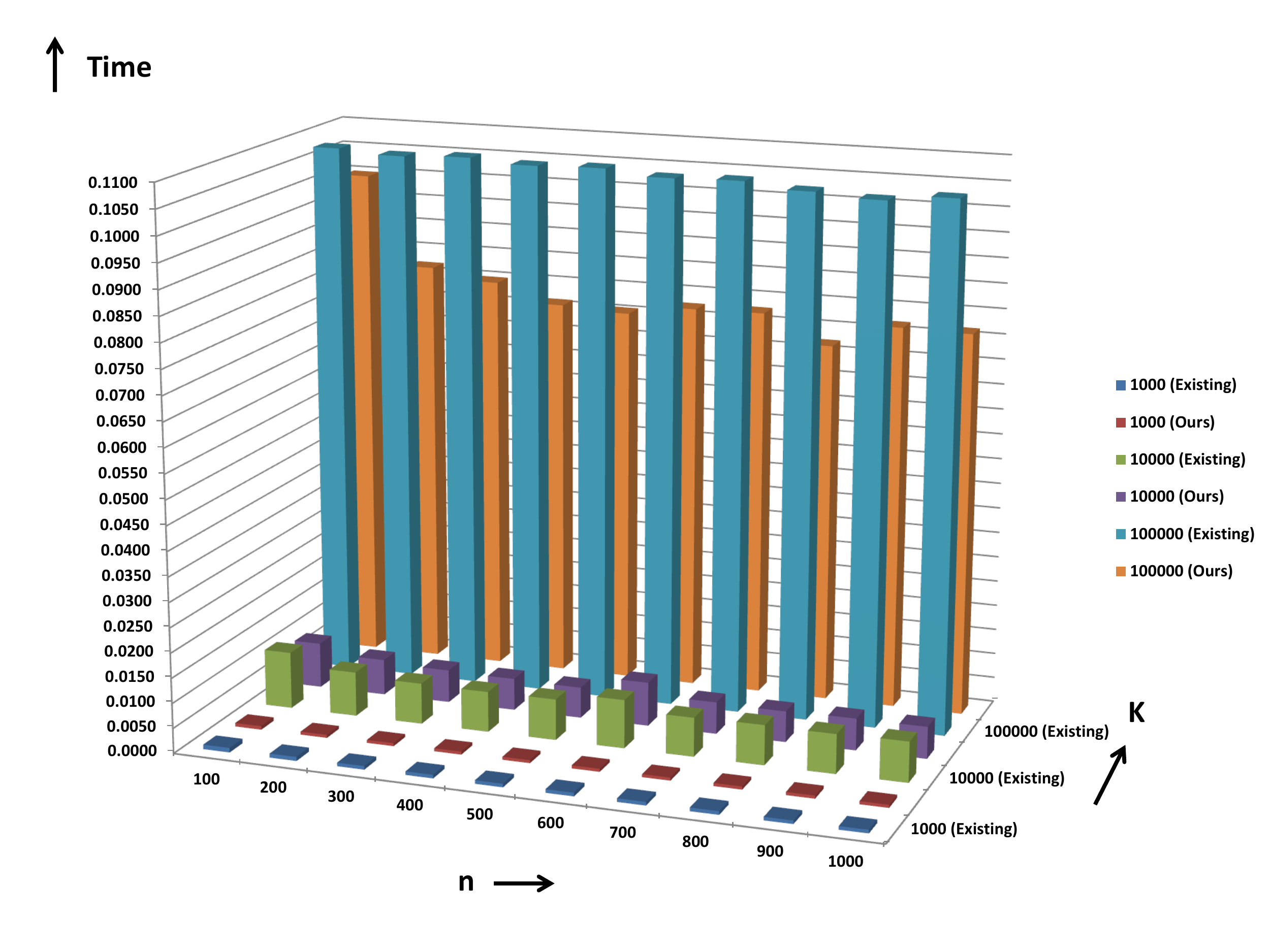}
\vspace*{-.9cm}
\caption{Variation of runtime for $k=1000, k=10000$ and $k=100000$}

\end{subfigure}% <-- added

\medskip
\centering % <-- added
\hspace*{-1 cm}
\begin{subfigure}{0.55\textwidth}
\includegraphics[width=\linewidth]{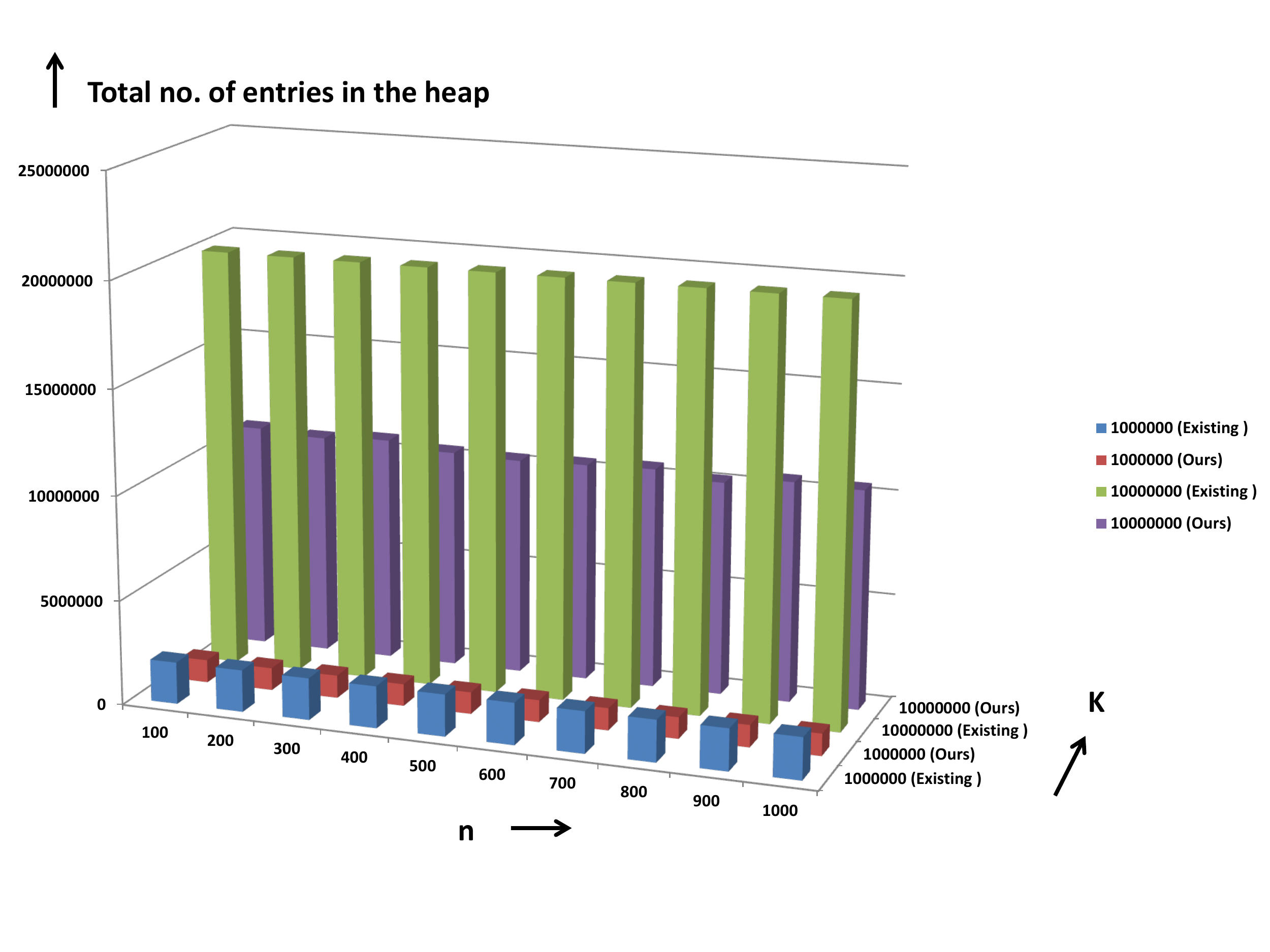}
\vspace*{-1.4cm}
\caption{Variation of total no. of heap entries for $k=1000000$ and $k=10000000$}
\end{subfigure}\hspace*{.3 cm}% <-- added
\begin{subfigure}{0.55\textwidth}
\includegraphics[width=\linewidth]{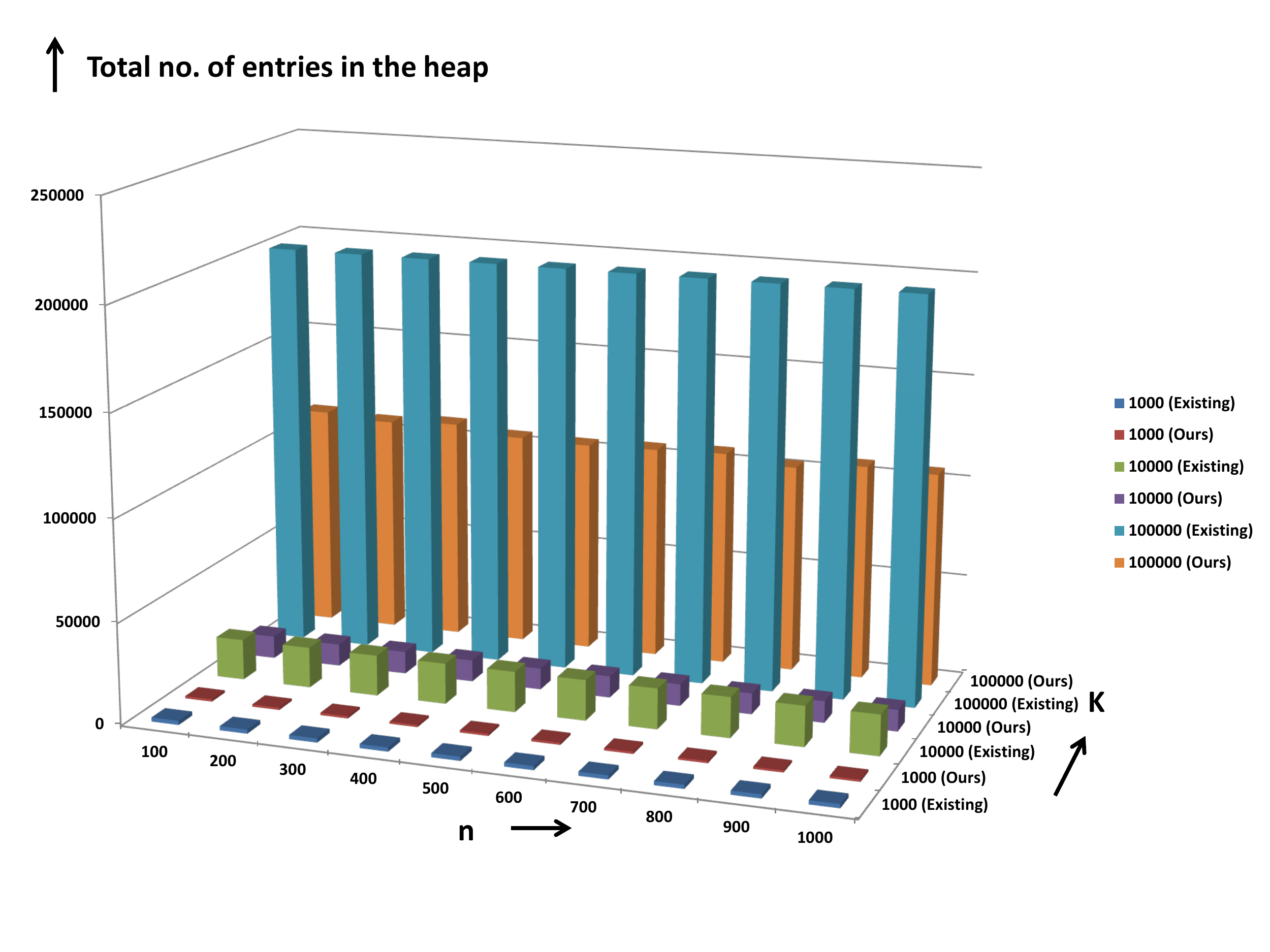}
\vspace*{-1.4cm}
\caption{Variation of total no. of heap entries for $k=1000, k=10000$ and $k=100000$}
\end{subfigure} % <-- added

\medskip
\centering % <-- added
\hspace*{-1 cm}
\begin{subfigure}{0.55\textwidth}
\includegraphics[width=\linewidth]{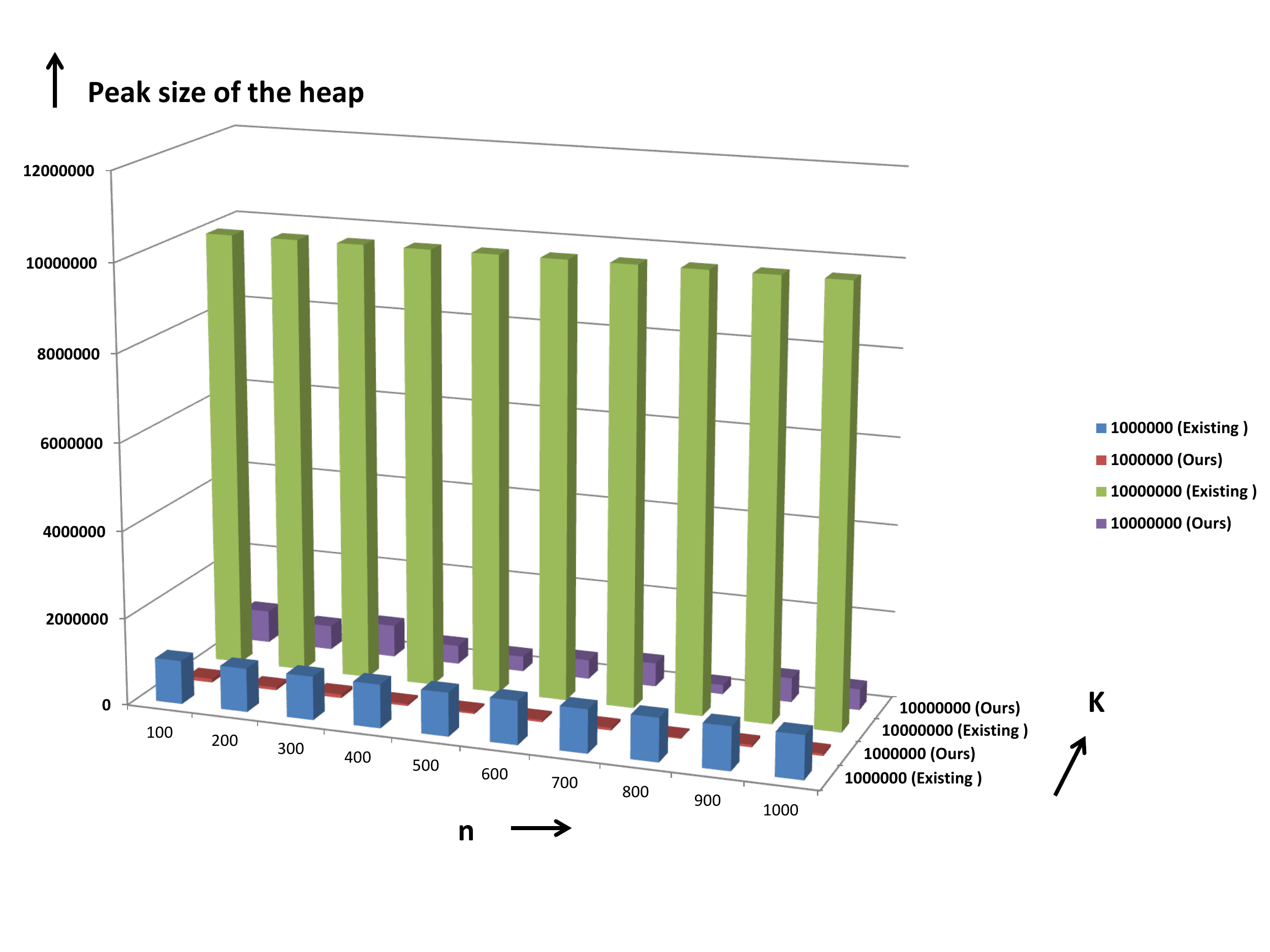}
\vspace*{-1.4cm}
\caption{Variation of peak size of the heap for $k=1000000$ and $k=10000000$}
\end{subfigure}\hspace*{.3 cm}% <-- added
\begin{subfigure}{0.55\textwidth}
\includegraphics[width=\linewidth]{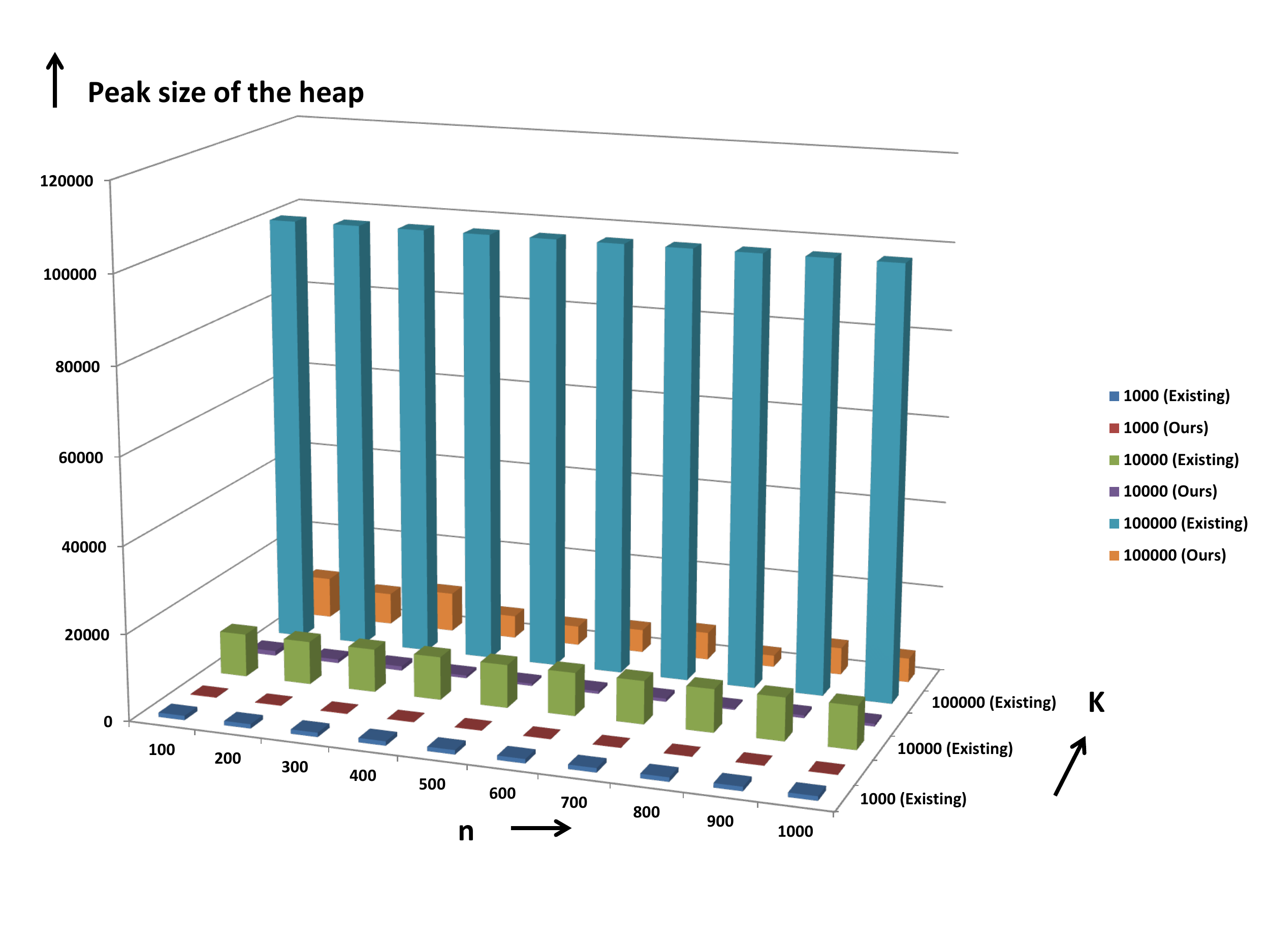}
\vspace*{-1.4cm}
\caption{Variation of peak size of the heap for $k=1000, k=10000$ and $k=100000$}

\end{subfigure}

\caption{Plots of comparing our algorithm with an existing algorithm reporting only the sums}
\label{fig:8}
\end{figure}

Several graphs for varying values of $k$ and $n$ are plotted. In Figure~\ref{fig:6}, the variation of run-time is plotted with $n$, keeping $k$ fixed for four different values and in Figure~\ref{fig:7}, the growth of run-time is plotted with respect to $k$, for different values of $n$. Note that the plotted graphs in Figure~\ref{fig:6} and Figure~\ref{fig:7} are very much in agreement with the complexity results stated in the lemma~\ref{algo2:time} and~\ref{algo_sum}, the runtimes for the modified algorithm being clearly independent of $n$. 

Finally, the results of our algorithm that only reports the sum of the subsets according to the ranks are compared with an existing algorithm~\cite{d_math15} that also achieves the same feat. Both the algorithms were implemented in C and were run on the same machine mentioned above and on exactly the same set of randomly generated integers. The running times, total number of entries into the heap and also the peak number of entries in the heap attained during the whole run for each of the two algorithms are reported in Table~\ref{tab2}. The gains are much more for the two latter parameters; the difference in runtime being somewhat compensated by the constant time updating of some integers values stored in each node, which is little more in our case. For the ease of visual comparison, results are displayed in the bar-charts given in Fig.~\ref{fig:8}. It can be seen that our method consistently performs better than the existing method.

\subsection{Analysis of the results}
It can be recalled that in the earlier work~\cite{d_math15}, each time a node is extracted from the heap, two of its children are inserted immediately and hence the number of total entries in the heap has a fixed value of $2k+1$. Also, the peak number of entries into the heap always rises to $k+1$, which can be verified from Table~\ref{tab2} also. This is because each extraction from the heap is followed by a double insertion that effectively increases the heap size by one after each instruction. The last insertion could be avoided however if the check is performed before insertion that whether $k^{th}$ item has been extracted. In that case, the numbers would have been reduced to $2k - 1$ and $k$ respectively.

Note that in our solution, every node of the DAG $G$ can have at most two children (Lemma~\ref{lemma_two_child}). However, most of the times, the number of children is only one and only in a few cases it is two and in some other cases it is zero also. Hence, though the number of entries in the heap of our solution lies in $[k, 2k-1]$, in practice, the actual number of entries turns out to be far less than $2k - 1$. Also, the peak number of entries into the heap at any point of time remains far lower than the value of $k$. This can be easily observed in Table~\ref{tab2}. It is sort of obvious because in our case, most of the times only one child node needs to be inserted in the heap after extracting the node with the minimum value from the heap. This could have been predicted very easily from the sparse nature of the graph $G$ shown in Fig.~\ref{fig:5}, that we implicitly kept on constructing during the run of the algorithm. It hardly increases the heap size as in most steps the extraction compensates the single insertion that follows. Actually, our algorithm exploits the partial order that inherently exists in the data in a better manner than the simpler algorithm. It is clear from the above discussion that due to the lesser value of the peak number of entries into the heap at any point of time; our approach will perform much better than the earlier solution if this is being applied in any practical problem, especially when there is a lot of satellite data in each node.

\section{Conclusion}
\label{section:conclusion}
An algorithm is proposed to compute the top-$k$ subsets of a set $R$ of $n$ real numbers, creating portions of an implicit DAG on demand, that gets rid of the storage requirement of the preprocessing step altogether. In several steps, we have made the DAG as sparse as possible so that the overall run-time complexity improves retaining its useful properties. Our algorithms were implemented in C and the plots of run-time illustrate that the algorithm is performing as expected. Another efficient algorithm is proposed for reporting only the top-$k$ subset sums (not subsets) and we have compared our results with an existing solution. These two algorithms were also implemented and the results show that our method is consistently performing better than the existing one. The proposed methodology actually has the rationale of better exploiting the partial order that is inherent in the structure of the problem and we feel it can be used in other similar problems also albeit with a few modifications. Solving the problem for aggregation functions other than summation, and finding other applications of the {\em metadata structure} remain possible directions for future research.

\section*{Acknowledgement}
\label{ack}
We sincerely acknowledge the useful suggestions given by Prof. Subhash Chandra Nandy of ACM Unit, Indian Statistical Institute, Kolkata, India, which helped us to improve the presentation of this manuscript.

\end{document}